\documentclass[sigconf]{acmart}

\usepackage{booktabs} 
\usepackage{amsmath}
\usepackage{amsfonts}
\usepackage[english]{babel}
\usepackage{tikz}
\usepackage[utf8]{inputenc}
\usepackage{calc}
\usepackage{mathtools}

\usepackage{algorithm2e}
\usepackage{subcaption}

\usepackage[subject={Todo},author={Nicolas}]{pdfcomment}

\usepackage{hyperref}

\setcopyright{rightsretained}

\theoremstyle{definition}
\newtheorem*{remark}{Remark}

\newcommand{\E}{\mathbb{E}}
\newcommand{\R}{\mathbb{R}}
\renewcommand{\P}{\mathbb{P}}

\renewcommand\footnotetextcopyrightpermission[1]{} 
\begin{document}

\title{On the Power-of-d-choices with Least Loaded Server Selection}

\author{T. Hellemans and B. Van Houdt}

\affiliation{Dept. Mathematics and Computer Science\\
University of Antwerp, Belgium}
\email{{tim.hellemans,benny.vanhoudt}@uantwerpen.be}

\begin{abstract}
Motivated by distributed schedulers that combine the power-of-d-choices with late binding and systems that use replication with cancellation-on-start, we study the performance of the LL(d) policy which assigns a job to a server that currently has the least workload among d randomly selected servers in large-scale 
homogeneous clusters.

We consider general service time distributions and propose a partial integro-differential equation to describe the evolution of the system. This equation relies on the earlier proven ansatz for LL(d) which asserts that the workload distribution of any finite set of queues becomes independent of one another
as the number of servers tends to infinity. Based on this equation we propose a fixed point iteration for the limiting workload distribution and
study its convergence.

For exponential job sizes we present a simple closed form expression for the limiting workload distribution that is valid for any work-conserving 
service discipline as well as for the limiting response time distribution in case of first-come-first-served scheduling.
We further show that for phase-type distributed job sizes the limiting workload and response time distribution can be expressed via 
the unique solution of a simple set of ordinary differential equations.

Numerical and analytical results that compare response time of the classic power-of-d-choices algorithm and the LL(d) policy are also presented
and the accuracy of the limiting response time distribution for finite systems is illustrated using simulation. 

\end{abstract}


\maketitle

\section{Introduction}
\label{sec:introduction}

Load balancing plays a crucial role in achieving low latency in large-scale clusters. A simple randomized approach, denoted as SQ(d), exists in assigning incoming jobs to a server that currently holds the fewest number of jobs among a set of $d$ randomly selected servers, the so-called {\it power-of-d-choices} algorithm \cite{vvedenskaya3,mitzenmacher2,Mukherjee2016,aghajani2017}. While this approach yields short queues with high probability in case of first-come-first-served (FCFS) scheduling even for general job size distributions provided that $d$ is chosen sufficiently large \cite{bramsonLB,bramsonAAP}, short queues do not guarantee low latency as the
queue length is only a coarse indicator of the waiting time in the presence of high job size variability. The main issue is that under the FCFS discipline 
short jobs can get stuck behind a single long job which significantly increases the short job latency. In addition when multiple dispatchers are
used to distribute the jobs, {\it race} conditions may occur where multiple schedulers concurrently place jobs on a server that appears lightly loaded
\cite{mitzenmacherOLD}.

To avoid these issues the notion of {\it late binding} was recently introduced in \cite{Sparrow}. With late binding the dispatcher still probes $d$ servers
at random, but the servers do not immediately reply by sending their queue length information. Instead they place a reservation at the end of a local work queue and when the reservation
reaches the front of the queue, the server requests the job associated to the reservation from the dispatcher. In this manner the job is assigned to the server
that is able to launch the job the soonest among the $d$ randomly selected servers. The downside of late binding is that the server always experiences some
idle time in between the execution of two jobs, which implies some efficiency loss. However, whenever the network latencies are much smaller than the shortest job runtimes (and the system load is not extremely high), experiments on a 110-machine cluster show that a scheduler that relies on late binding performs close to an ideal scheduler \cite{Sparrow}.  

Note that late binding as described above is equivalent to assigning the job to the server that has the least workload among
$d$ randomly selected servers, which is known as the LL(d) policy \cite{BramsonLB_Questa}, provided that the network latencies are negligible\footnote{When the network latencies are not negligible compared to the job runtimes, we can regard them as part of the workload of a job such that the job execution consists of two parts: fetching the job and executing it, see Section \ref{sec:overhead}.}.

The main objective of this paper is to study the large-scale limit of the server workload and response time distribution of the LL(d) policy
when employed on a homogeneous cluster subject to Poisson job arrivals with general service times. For this purpose we introduce a partial
integro-differential equation that captures the evolution of the so-called {\it cavity process} and study its equilibrium.  
The key observation, established in \cite{BramsonLB_Questa}, that is under the LL(d) policy with general service time distributions,
the workload distribution of any finite set of servers becomes asymptotically independent as the number of servers tends to infinity
(provided that all the servers employ the same local non-idling service discipline, e.g., FCFS, PS, etc.). 
Moreover, the limit of the marginal workload distribution of a server corresponds to the unique equilibrium environment.

It is worth noting that the LL(d) policy is equivalent to the following system that uses replication with cancellation-on-start 
to reduce waiting times. Arriving jobs are replicated $d$ times and are randomly assigned to $d$ servers (that all operate in FCFS order).
As soon as a single replica starts execution on a server, the remaining $d-1$ replicas are killed (with the additional
assumption that if multiple replicas start at exactly the same time, only one is executed). Prior work on replication 
was mainly done in the context of systems that experience server slowdown and therefore focused on
replication with cancellation-on-job-completion \cite{gardnerSIGM,gardnerOR}, which is considerably different 
from LL(d) as jobs are often (partially) executed on multiple servers in such case. 

Another reason for studying the large-scale limit of the LL(d) policy exists in understanding how much benefit precise workload information 
gives in comparison to the coarser queue length information used by SQ(d).

The main contributions of the paper are as follows:
\begin{enumerate}
\item A partial integro-differential equation to describe the
transient evolution of the limiting workload of a server under the LL(d) policy is derived.
\item An integral equation for the limiting stationary workload distribution is presented together
with a fixed-point iteration to compute its solution. Convergence of the fixed-point
iteration is proven for $\rho < e^{-1/e} \approx 0.6922$.   
\item A simple explicit solution for the limiting workload and response time distribution
is presented in case of exponential job sizes. For phase-type distributed job sizes we prove
that the limiting workload distribution can be computed easily by solving a simple set of 
ordinary differential equations.
\item We present both analytical and numerical results that compare the response time of the
LL(d) policy with the classic SQ(d) policy. These results illustrate that late binding offers a
significant reduction in the response time under a very wide range of loads even when  taking the
idleness caused by late binding into account. 
\end{enumerate}
The paper is structured as follows. The model considered in this paper is described in Section \ref{sec:model}.
The partial integro-differential equation that captures the transient evolution of the workload is
introduced in Section \ref{sec:cavity}, while the integral equation for the limiting stationary workload 
and its associated fixed point equation are presented in Section \ref{sec:workload}. Sections \ref{sec:expo} and \ref{sec:ph}
discuss the special cases of exponential and phase-type distributed job sizes, respectively.
Section \ref{sec:versus} compares the performance of the LL(d) and SQ(d) policies, while Section \ref{sec:finite}
briefly studies the accuracy of the limiting distributions for systems of finite size.
Conclusions are drawn in Section \ref{sec:concl}.

\section{Model description}\label{sec:model}
We consider a system consisting of $N$ single server queues each having an infinite waiting room. Arrivals occur into
the system as a Poisson process with rate $\lambda N$. For each incoming customer 
$d$ queues are selected uniformly at random (with replacement) and the job joins the queue that currently holds the least 
workload with ties being broken uniformly at random. The service discipline is such that the workload
at any queue reduces at rate $1$ when positive, that is, we do not put any restriction on the service discipline apart
from the fact that it is non-idling and identical in each server (unless stated otherwise). 
The workload offered by a job has a general distribution with cdf $G(\cdot)$, pdf $g(\cdot)$, 
mean $\E[G]$ and is such that $G(0)=0$. We define $\rho = \lambda \E[G]$ and assume that $\rho < 1$.

The above model corresponds to the so-called least-loaded supermarket model, denoted as LL(d) in \cite{bramsonLB,BramsonLB_Questa}.
Note that the corresponding Markov process that keeps track of the workloads of the $N$ queues is positive Harris recurrent and has a 
unique stationary probability measure $\mathcal{E}^{(N)}$ whenever the queueing system is subcritical, that is, when $\rho  < 1$, as noted at 
the end of Section 5 in \cite{bramson2011}. In fact, this result is as a special case of \cite[Theorem 2.5]{Foss98}.

\section{Cavity process}\label{sec:cavity}

We start by introducing the cavity process from \cite{BramsonLB_Questa} for the LL(d) supermarket model. 
The  process is intended to capture the evolution
of the workload of a single queue for the limiting system where the number of servers $N$ tends to infinity.

\begin{definition}[LL(d) cavity process]
Let $\mathcal{H}(t)$, $t \geq 0$, be a set of probability measures on $\mathbb{R}$ called the {\it environment process}. 
The {\it cavity process} $X^{\mathcal{H}(\cdot)}(t)$, $t \geq 0$, takes values in $\mathbb{R}$ and is defined as follows. 
Potential arrivals occur according to a Poisson process with rate $\lambda d$. When a potential arrival
occurs at time $t$, we compare the state $X^{\mathcal{H}(\cdot)}(t-)$ just prior to time $t$ with 
the states of $d-1$ independent random variables with law $\mathcal{H}(t)$.
The potential incoming job is assigned to the state among these $d$ states that has the lowest value, where ties are
broken uniformly at random. If the job is assigned to state $X^{\mathcal{H}(\cdot)}(t-)$, we immediately
add the job to the queue, that is, $X^{\mathcal{H}(\cdot)}(t)=X^{\mathcal{H}(\cdot)}(t-)+x$ where $x$ is the
size of the incoming job. Otherwise, the job immediately leaves the system, i.e., $X^{\mathcal{H}(\cdot)}(t)=X^{\mathcal{H}(\cdot)}(t-)$.
Clearly, if $X^{\mathcal{H}(\cdot)}(t-)$ has law  $\mathcal{H}(t)$ a potential arrival at time $t$ joins the queue
with probability $1/d$. Finally, the cavity process decreases at rate one during periods without arrivals and is lower
bounded by zero.
\end{definition}

\begin{definition}[Equilibrium Environment]
When a cavity process $X^{\mathcal{H}(\cdot)}(\cdot)$ has distribution 
$\mathcal{H}(t)$ for all $t \geq 0$, we say that $\mathcal{H}(\cdot)$ is an {\it equilibrium environment process}. 
Further, a probability measure $\mathcal{H}$ is called an {\it equilibrium environment} if 
$\mathcal{H}(t) = \mathcal{H}$ for all $t$ and $X^{\mathcal{H}(\cdot)}(t)$ has distribution $\mathcal{H}$ for all $t$.
\end{definition}

\begin{theorem}[due to Theorem 2.2 of \cite{BramsonLB_Questa}]
Consider the LL(d) supermarket model with $N$ queues, general service times (with mean E[G]),
Poisson arrivals with rate $\lambda N < N/E[G]$
and an identical non-idling service discipline at each queue. 
Let $\mathcal{E}^{(N,N')}$ be the projection of the stationary measure $\mathcal{E}^{(N)}$ of the $N$ workloads into the 
workloads of the first $N'$ queues,
then $\mathcal{E}^{(N,N')}$ converges in total variation to the $N'$-fold convolution of
$\mathcal{E}^{(\infty,1)}$ (in an appropriate metric space) as $N$ tends to infinity. 
Moreover, $\mathcal{E}^{(\infty,1)}$ is the unique equilibrium environment of the LL(d) supermarket model.
\end{theorem}

In other words the above theorem indicates that the workload distributions of any finite set of
$N'$ queues becomes asymptotically independent as $N$ tends to infinity and the marginal workload
distribution of any queue is given by the {\it unique} equilibrium environment $\mathcal{H}$ of the
LL(d) supermarket model. 

We now characterize the evolution of the cavity process associated with the
equilibrium environment process $\mathcal{H}(\cdot)$ of the LL(d) supermarket model.

Let $f(t,s)$ for $s \in \mathbb{R}^+_0 := (0,\infty)$ describe the density of servers which, at time $t$, have workload $s$. Note that $f(t,\cdot)$ is not a real probability density function (pdf) as some of the servers may be idle, denote $F(t,0) := 1 - \int_{0}^{\infty} f(t,s) ds$ (where $f(t,0)$ 
may be defined arbitrarily). In the following we will refer to $f(t,\cdot)$ as a density, and we define its cumulative distribution function (cdf) $F(t,\cdot)$ as $F(t,s) = F(t,0) + \int_{0}^s f(t,u) du$.

For any $d \in \{2,3,\dots\}$, we define the function $c_d(t,u)$ as the density at which a potential arrival at time $t$ joins the cavity
queue with workload $u > 0$.  By definition of the cavity process
associated to the equilibrium environment, this density is given by:
\begin{align}\label{eq:C}
c_d(t,u) &= f(t,u)(1-F(t,u))^{d-1} = f(t,u)\bar{F}(t,u)^{d-1},
\end{align}
where we use the notation $\bar{F}(t,u) = 1 - F(t,u)$ for the complementary cdf (ccdf).
We further denote the probability that a potential arrival at time $t$ joins the cavity queue with workload at most $u$ by $C_d(t,u)$. 
In this case we have, as ties are broken uniformly at random:
\begin{align}\label{eq:barC}
C_d(t,u) &= F(t,0) \sum_{k=0}^{d-1} \binom{d-1}{k }\frac{F(t,0)^k \bar{F}(t,0)^{d-1-k}}{k+1} 
+ \int_{v=0}^u  c_d(t,v) dv 
\nonumber \\
&= \frac{1 - \bar{F}(t,0)^d}{d} + \int_{v=0}^u  c_d(t,v) dv  = \frac{1 - \bar{F}(t,u)^d}{d}.
\end{align}
In particular, $C_d(t,0)$ is the probability that a potential arrival joins an empty cavity queue.

\begin{theorem}
The evolution of the cavity process associated to the equilibrium environment of the LL(d)
supermarket model is captured by the following set of equations: 
\begin{align}\label{eq:PIDE}
&\frac{\partial f(t,s)}{\partial t} - \frac{\partial f(t,s)}{\partial s} = \lambda d\int_0^s c_d(t,u) g(s-u) du \nonumber \\ 
&\hspace*{2cm} + \lambda d C_d(t,0) g(s) - \lambda d c_d(t,s) \\
&\frac{\partial F(t,0)}{\partial t} = f(t,0^+) - \lambda d C_d(t,0),
\label{eq:ODEF0}
\end{align}
for $s>0$, where $f(x,z^+) = \lim_{y \downarrow z} f(x,y)$.
\end{theorem}
\begin{proof}
Assume $s > 0$ and let $s > \Delta > 0$ be arbitrary. 
In order to have a workload of $s$ at time $t+ \Delta$ we need to consider three possible cases:
no arrivals in $[t,t+\Delta]$, an arrival occurs in $[t,t+\Delta]$ when the workload is non-zero and
an arrival occurs in an idle server in $[t,t+\Delta]$. Hence, we can write     
\begin{equation}\label{eq:f1}
f(t+\Delta, s) = Q_1 + Q_2 + Q_3.
\end{equation}
The terms $Q_i$, for $i=1,2$ and $3$ are discussed next.
\begin{enumerate}
\item[1)] No arrivals in the interval $[t,t+\Delta]$: if the cavity queue at time $t$ has a workload exactly equal to $s + \Delta$ and has no arrivals 
in $[t,t+\Delta]$, it will have a workload equal to $s$ at time $t+\Delta$. The density of having a workload $s+\Delta$ at time $t$ is given by $f(t,s+\Delta)$ and the density at which an arrival occurs at the cavity queue at time $t+v, v \in [0,\Delta]$, when it has workload $s+\Delta-v$,
is equal to $\lambda d c_d(t+v, s + \Delta - v)$. Therefore we find:
$$
Q_1 = f(t, s + \Delta) - \lambda d \int_{v=0}^{\Delta} c_d(t+v,s+\Delta - v) dv  + o(\Delta).
$$
\item[2)]A single arrival occurs when the cavity queue is not idle: in this case at some time $t+v, v \in [0,\Delta]$ an arrival of size $s+\Delta-u$ at 
the cavity queue which has workload $u-v$ for some $u \in [v,s+\Delta]$ occurs. We find:
$$
Q_2 = \lambda d  \int_{v=0}^{\Delta} \int_{u=v}^{s+\Delta} c_d(t+v, u-v) g(s+\Delta-u)dudv  + o(\Delta).
$$
\item[3)] A single arrival occurs when the cavity queue is empty: in this case a job 
of size $s + \Delta - v$ arrives at time $t+v$ for some $v \in [0,\Delta]$. Hence,
$$
Q_3 = \lambda d \int_{v=0}^{\Delta} C_d(t+v, 0) g(s+\Delta-v)dv  + o(\Delta).
$$
\end{enumerate}
By subtracting $f(t,s+\Delta)$, dividing by $\Delta$ and letting  $\Delta$ decrease to zero, we find \eqref{eq:PIDE} from \eqref{eq:f1}.

We still require a differential equation for $F(t,0)$, a server may be idle at time $t$ by remaining idle in $[t,t+\Delta]$ or having a workload equal to $\Delta - v, v < \Delta$ at time $t + v$. We therefore find:
\begin{align*}
F(t+\Delta,& 0) = F(t,0) - \lambda d \int_{v = 0}^{\Delta}C_d(t+v,0)dv \\
&+ \int_{v=0}^{\Delta} f(t+v, \Delta - v) du + o(\Delta),
\end{align*}
subtracting $F(t,0)$, dividing by $\Delta$ and letting $\Delta$ tend to zero yields \eqref{eq:ODEF0}.
\end{proof}

\begin{remark}
The set of equations given by (\ref{eq:PIDE}-\ref{eq:ODEF0}) can be solved numerically using the following scheme:
\begin{align*}
f(t+\delta,0^+) &=
\lambda d C_d(t,0),\\
f(t+\delta,s) &=
f(t,s+\delta) + \lambda d \delta \int_0^s c_d(t,u) g(s-u) du \\
&  + \lambda d \delta  C_d(t,0) g(s) - \lambda d \delta  c_d(t,s),
\end{align*}
for $s \geq \delta$. As a boundary condition, we may impose that we start with all servers being idle, 
i.e., for $s>0$ we set $f(0,s) = 0$ and $F(0,0) = 1$.
We are however mainly interested in the long-term behavior of the model, i.e., as $t$ tends to infinity.
\end{remark}

\section{Limiting workload distribution}\label{sec:workload}
As indicated in the previous section, the limiting stationary workload distribution is given by the unique equilibrium environment.
Let $F(s)$ be the cdf of the workload distribution, that is, $F(s)$ represents the probability that the workload is at most
$s$ and let $f(s)$ be its density for $s > 0$. 
Furthermore, similar to \eqref{eq:C} and \eqref{eq:barC}, define
\begin{align}
c_d(u) =  f(u)\bar{F}(u)^{d-1},
\end{align}
and
\begin{align}
C_d(u) = \frac{1 - (1-F(u))^d}{d}.
\end{align}

\begin{theorem}
The stationary workload distribution is the unique distribution that obeys the following integral equation:
\begin{equation}\label{eq:F(s)}
F(s) = (1-\rho) + \lambda \cdot \left(
\int_{0}^s (1 - \bar{F}(u)^d) (1 - G(s-u)) du
\right)
\end{equation}
\end{theorem}
\begin{proof}
By demanding that the derivatives with respect to $t$ are zero in (\ref{eq:PIDE}-\ref{eq:ODEF0}), we find
\begin{align}\label{eq:PIDElim}
\frac{\partial f(s)}{\partial s}&=\lambda d \left(c_d(s) - \int_{0}^{s}c_d(u)\,g(s-u)\,du -C_d(0)\,g(s)\right),
\end{align}
and
\begin{align}\label{eq:PIDElim2}
f(0^+) = \lambda d C_d(0).
\end{align}
Integrating \eqref{eq:PIDElim} once (and relying on the assumption that $G(0) = 0$) we find:
\begin{align}
f(s)
=
K-\lambda d \cdot \left( \frac{1}{d}-C_d(s) + C_d(0) G(s) + \int_0^s c_d(u) G(s-u) du \right),
\end{align}
for an appropriate constant $K$. As we know from \eqref{eq:PIDElim2} that $f(0^+) = \lambda d C_d(0)$, we see that we should set $K$ equal to $\lambda$. We may therefore conclude that
\begin{equation}\label{eq:f(s)}
f(s) =
\lambda d \cdot
\left(
C_d(s) - C_d(0) G(s) - \int_0^s c_d(u) G(s-u) du
\right)
\end{equation}
Integrating equation \eqref{eq:f(s)} once more and using the fact that $F(0) = 1-\rho$, yields
\begin{align*}
F(s) = (1-\rho) + \lambda d  \cdot \left(
\int_{0}^s C_d(u) (1 - G(s-u)) du
\right)
\end{align*}
The uniqueness follows from the fact that there exists a unique equilibrium environment for the LL(d) supermarket model
as stated earlier.
\end{proof}

\begin{remark}
The cavity process evolves as the workload of an M/G/1 queue with a workload dependent arrival rate, we can therefore also
apply Theorem 2.1 in \cite{bekker2004} to the LL(d) cavity process. In this manner we obtain that
$$
f(s) = \lambda d \left( C_d(0) (1-G(s)) + \int_0^s c_d(u) (1-G(s-u)) du \right),
$$
which can easily be shown to be equivalent to \eqref{eq:f(s)} by using the fact that $c_d(u) = \frac{d}{du} C_d(u)$. 
The interpretation of this equation is as follows. The left-hand side of the equation corresponds to the downcrossing rate through level $s$, while the right-hand side denotes the upcrossing rate through $s$.
\end{remark}



\subsection{Fixed point iteration}\label{sec:fixed}
We propose to use the following simple fixed point iteration to solve the integral equation
\eqref{eq:F(s)}:
\[F_{n+1}(s) = (1-\rho) + \lambda \cdot \left(
\int_{0}^s (1 - \bar{F}_n(u)^d) (1 - G(s-u)) du
\right),\]
which we prove converges to the unique fixed point provided that $\rho < d^{-1/d}$.
In Section \ref{sec:ph} we further show that if the service time distribution is a phase-type
distribution, we can directly compute the limiting workload distribution $F(s)$ by solving a 
simple set of differential equations (for any $\rho < 1$), meaning there is no need to 
make use of the above fixed point iteration.  

Define the space $\operatorname{CDF}_{1-\rho} \subseteq [1-\rho,1]^{[0,\infty)}$ to be the space of cumulative distribution functions starting in $1-\rho$, i.e., the space of functions which satisfy: 
\begin{itemize}
\item $F(0) = 1-\rho$,
\item $\lim_{s\rightarrow\infty} F(s) = 1$,
\item for $s,h>0: F(s+h) \geq F(s)$,
\item $\lim_{h\rightarrow 0^+} F(s+h) = F(s)$.
\end{itemize}
On this space we can define an operator $T_d: \operatorname{CDF}_{1-\rho} \longrightarrow \R^{[0,\infty)}$ defined by:
$$
T_dF: [0,\infty) \rightarrow \R: s \mapsto (1-\rho) + \lambda d \cdot \left( \int_{0}^s C_d(u) (1 - G(s-u)) du \right).
$$
\begin{lemma}\label{operator_lem}
For $F \in \operatorname{CDF}_{1-\rho}$, we have $T_dF \in \operatorname{CDF}_{1-\rho}$.
\end{lemma}\begin{proof}
The only non-trivial part is to show that $\lim_{s\rightarrow \infty} T_dF(s) = 1$. We find:
\begin{align*}
\lim_{s\rightarrow \infty} \left| \int_0^s d C_d(u) \cdot (1 - G(s-u)) du \right|
&\leq \lim_{s\rightarrow \infty} \int_0^s (1 - G(s-u)) du\\
&= \E[G],
\end{align*}
which shows that $\lim_{s\rightarrow \infty} T_dF(s) \leq 1$. To obtain the other inequality observe that for any $\varepsilon > 0$, we can find a $U > 0$ for which:
\begin{align*}
&\lim_{s\rightarrow \infty} \int_U^s (1 - G(s-u)) du > \sqrt{1 - \varepsilon} \E[G],
&C_d(u) \geq \sqrt{1 - \varepsilon},
\end{align*}
for $u > U$.
We thus find:
\begin{align*}
\lim_{s\rightarrow \infty} \int_0^s d C_d(u) (1 - G(s-u))du
& \geq \lim_{s\rightarrow \infty} \int_U^s d C_d(u) (1 - G(s-u))du\\
& \geq (1 - \varepsilon)\E[G]
\end{align*}
this shows that $\lim_{s\rightarrow\infty}T_dF(s) \geq 1$
\end{proof}

\begin{remark}
Due to the above lemma we may write $T_d: \operatorname{CDF}_{1-\rho} \rightarrow \operatorname{CDF}_{1-\rho}$.
\end{remark}
\begin{remark}
We can define an order on $\mbox{CDF}_{1-\rho}$ by stating that $F_1 \preceq F_2 \Leftrightarrow \forall s \in [0,\infty): F_1(s) \leq F_2(s)$, then a simple application of the Knaster-Tarski theorem also guarantees the existence of a fixed point of $T_d$. Indeed note that we have $F_1 \preceq F_2 \Rightarrow T_d F_1 \preceq T_d F_2$.
\end{remark}

\begin{theorem}\label{operator_eig}
For any $F_1,F_2 \in \operatorname{CDF}_{1-\rho}$ we have:
$$
d_K(T_dF_1,T_dF_2) \leq d \rho^d \cdot d_K(F_1, F_2),
$$
where $d_K$ denotes the uniform (or Kolmogorov) metric, i.e.,
$d_K(F_1, F_2) = \sup_s |F_1(s)-F_2(s)|$.
\end{theorem}
\begin{proof}
Let $\varepsilon > 0$ be arbitrary and let $s^*$ be such that:
\begin{align*}
\sup_s \int_0^s |(1-F_1(u))^d - (1-F_2(u))^d| (1 - G(s-u)) du\\
<
\int_0^{s^*} |(1-F_1(u))^d - (1-F_2(u))^d| (1 - G(s^*-u)) du + \varepsilon.
\end{align*}
We therefore have that $d_K(T_dF_1,T_dF_2)$ is bounded above by:
\begin{align*}
\lambda \int_0^{s^*} |(1 - F_2(u))^d - (1-F_1(u))^d| (1 - G(s^*-u)) du + \varepsilon,
\end{align*}
We now use the fact (which can be shown by applying the mean value theorem) that for any $x,y \in [0,\rho)$
we have $|x^d - y^d| \leq d  \rho^{d-1} \cdot |x-y|$. This shows by applying the above that we have:
\begin{align*}
d_K(T_dF_1,T_dF_2)
&
< \lambda \int_0^{s^*} d \rho^{d-1} |F_1(u) - F_2(u)| (1 - G(s^*-u)) du + \varepsilon\\
&
\leq \lambda d \rho^{d-1} d_K(F_1,F_2) \int_0^{s^*} (1 - G(s^*-u)) du + \varepsilon\\
&\leq d \rho^{d} d_K(F_1,F_2) + \varepsilon,
\end{align*}
which completes the proof.
\end{proof}

\begin{remark}
In particular for $\rho < e^{-1/e} \approx 0.6922$ the above theorem shows by the Banach fixed-point theorem that $T_d$ admits a unique fixed point which can be found by our proposed fixed point iteration with speed of convergence $d_K(F^*, F_n) \leq \frac{d^n \rho^{nd}}{1 - d \rho^d} d_K(F_1,F_0)$. 
This follows from the fact that $d^{-1/d}$ attains a minimum in $e$. For higher values of 
$\rho$, $d$ must be such that $d \rho^d < 1$ to guarantee convergence via Theorem \ref{operator_eig}.
Numerical experiments using both light-tailed and heavy-tailed distributions suggest that the fixed point
iteration converges quickly for any $\rho < 1$. 
\end{remark}

\section{Exponential Job Sizes}\label{sec:expo}
In the previous section we established an integral equation for the limiting stationary workload
distribution (for any non-idling service discipline). In this section we derive an explicit
expression for this distribution in case of exponential job sizes with mean $1$, that is, when $G(s) = 1 - e^{-s}$
and $\rho = \lambda$.
In addition we also derive an explicit expression for the limiting response time distribution in case
the service discipline is first-come-first-served.

\subsection{Limiting workload distribution}

\begin{theorem}
The ccdf of the limiting stationary workload distribution for the $LL(d)$ policy for any non-idling 
service discipline with exponential job sizes with mean $1$ is given by:
\begin{equation}\label{eq:expExact}
\bar{F}(s) = (\lambda + (\lambda ^{1-d} - \lambda) e^{(d-1)s})^{\frac{1}{1-d}}.
\end{equation}
\end{theorem}
\begin{proof}
Using \eqref{eq:F(s)} with $G(s) = 1 - e^{-s}$
and $\rho = \lambda$, we have
\begin{equation}\label{eq:F(s)exponential}
F(s) = (1-\lambda) + \lambda d \int_{0}^s C_d(u) e^{u-s} du,
\end{equation}
Taking the derivative on both sides and using Leibniz integral rule, we find the following simple ODE for $F(s)$:
\begin{align}\label{eq:ODEexp}
F'(s) &= \lambda (1-\bar{F}(s)^d) -\lambda \int_{0}^s (1-\bar{F}(u)^d)e^{u-s}du \nonumber\\
&= \lambda (1-\bar{F}(s)^d) - (F(s)-(1-\lambda)) \nonumber \\
&= \bar{F}(s) - \lambda \bar{F}(s)^d,
\end{align}
with boundary condition $F(0) = 1-\lambda$, equivalently:
$$
\bar{F}'(s) = \lambda \bar{F}(s)^d - \bar{F}(s), 
$$
with $\bar{F}(0) = \lambda$.
This ODE can be solved explicitly and one easily verifies that the solution $\bar{F}(s)$ is given by:
$$
\bar{F}(s) = (\lambda + (\lambda ^{1-d} - \lambda) e^{(d-1)s})^{\frac{1}{1-d}}.
$$
\end{proof}

\begin{remark}
There is a striking and unexpected similarity between the limiting workload distribution of the LL(d) policy
and the response time distribution of the replication with cancellation-on-completion \cite[Section 5]{gardnerOR} in case of exponential
job sizes in the sense that the response time distribution of the latter system solves exactly
the same ODE as in \eqref{eq:ODEexp}, except that it is subject to the boundary condition $\bar{F}(0) = 1$.
\end{remark}

\begin{remark}
As $d$ tends to infinity, $\bar{F}(s)$ tends to $\lambda e^{-s}$
as $(\lambda^{1-d} - \lambda) ^{1/(1-d)}$ tends to $\lambda$.
This result is expected as for large $d$ we expect that a fraction $\lambda$ of the
servers contains exactly one job and the remaining workload of any such job is
exponentially distributed due to the memoryless nature of the exponential distribution.
\end{remark}



In order to obtain an expression for the expected workload of a server, we first recall the following 
integral representation for the analytic continuation of the hypergeometric function $\prescript{}{2}{F}_1(a,b;c;z)$ 
\cite[Chapter 15]{abramowitz64}
\begin{align}\label{eq:int2F1}
\prescript{}{2}{F}_1(a,b;c;z) = \frac{1}{B(b,c-b)} \int_{0}^1 x^{b-1} (1-x)^{c-b-1} (1-zx)^{-a} dx,
\end{align}
where $B(x,y) = \int_0^1 t^{x-1} (1-t)^{y-1} dt$ is the Beta function. 
This integral expression is valid for any $c > b > 0$ and $z < 1$. When $|z| < 1$ this function can be represented
as an infinite sum using the Pochhammer symbol (or falling factorial) $(q)_n = \prod_{k=0}^{n-1} (q+k)$ when
$n > 0$ and $(q)_0 = 1$: 
\begin{align}\label{eq:sum2F1}
\prescript{}{2}{F}_1(a,b;c;z) = \sum_{n=0}^\infty \frac{(a)_n(b)_n}{(c)_n} \frac{z^n}{n!}.
\end{align}

\begin{theorem}\label{th:Wd}
The mean $W_d(\lambda)$ of the limiting workload distribution of a server under the $LL(d)$ policy with exponential job sizes with mean $1$
is given by:
\begin{equation}\label{W_d}
W_d(\lambda) = \sum_{n=0}^{\infty} \frac{\lambda^{dn+1}}{1 + n(d-1)},
\end{equation}
in particular we find:
\begin{align*}
W_2(\lambda) &= - \frac{\log\left( 1 - \lambda ^ 2 \right)}{\lambda},\\
W_3(\lambda) &= - \frac{1}{\sqrt{\lambda}} \cdot \log\left( \frac{\sqrt{1 - \lambda^3}}{\lambda^{3/2} + 1} \right).
\end{align*}
\end{theorem}
\begin{proof}
We employ the notation $b = \lambda^{1-d} - \lambda$. 
We begin by computing (using $y=e^{-s}$ and $x=y^{d-1}$):
\begin{align*}
W_d(\lambda) &= \int_0^\infty \bar{F}(s)ds\\
&=  \int_0^1 \frac{1}{(\lambda y^{d-1} + b)^{1/(d-1)}} dy\\
&=  \frac{1}{b^{1/(d-1)}} \frac{1}{(d-1)} \int_0^1 \frac{x^{-(d-2)/(d-1)}}{(1+\frac{\lambda}{b} x)^{1/(d-1)}} dx
\end{align*}
Hence, by \eqref{eq:int2F1} this last integral can be expressed via the hypergeometric function $\prescript{}{2}{F}_1$ as 
$$
W_d(\lambda) = \frac{1}{b^{1/(d-1)}}\cdot \prescript{}{2}{F}_1\left( \frac{1}{d-1}, \frac{1}{d-1}; 1 + \frac{1}{d-1}; -\frac{\lambda}{b}\right).
$$
Note that we cannot directly use the sum representation of $_2F_1$ as $\lambda / b$ may become greater than $1$ (which happens when 
$\lambda$ gets close to one). Therefore we now employ the well-known linear transformation formulas:
\begin{align}
\prescript{}{2}{F}_1(a,b;c;z) & = (1-z)^{c-a-b} \cdot  \prescript{}{2}{F}_1(c-a, c-b; c; z) \nonumber\\
\prescript{}{2}{F}_1(a,b;c;z) & = (1-z)^{-a} \cdot  \prescript{}{2}{F}_1\left(a,c-b;c;\frac{z}{z-1}\right).
\label{eq:lintrans2}
\end{align}
Using these indicates that
\begin{align*}
W_d(\lambda) & = \frac{1}{b^{1/(d-1)}} \left(1+\frac{\lambda}{b}  \right)^{-\frac{1}{d-1}} \cdot  \prescript{}{2}{F}_1\left( 1,\frac{1}{d-1};1+\frac{1}{d-1};\lambda^d \right) \\
& = \lambda \cdot \prescript{}{2}{F}_1\left( 1,\frac{1}{d-1};1+\frac{1}{d-1};\lambda^d \right) 
\end{align*}
As $\lambda^d \in (0,1)$, we can use the sum representation given by \eqref{eq:sum2F1} to find that
$$
W_d(\lambda) =\sum_{n=0}^{\infty} \frac{\lambda^{nd+1}}{1 + n (d-1)},
$$
as $(1)_n = n!$ and $(1/(d-1))_n/(1/(d-1)+1)_n = 1/(1+n(d-1))$.
The expressions for $d=2,3$ can be either found directly by looking at the Taylor expansion of the logarithm or by solving the integral representation of $W_d(\lambda)$.
\end{proof}

\subsection{Limiting response time distribution}
We now focus on the limiting response time distribution $R$ in case the service discipline is first-come-first-served 
and denote its cdf as $F_R(s)$.

\begin{theorem}
The ccdf of the limiting response time distribution of the LL(d) policy with
FCFS service and exponential job sizes with mean $1$ is given by:
\begin{equation}\label{eq:FR}
\bar{F}_R(s) = \left( \lambda^d + (1 -\lambda^d) e^{(d-1)s} \right)^{\frac{1}{1-d}}.
\end{equation}
\end{theorem}
\begin{proof}
Let $E$ be an exponential random variable with mean $1$ and let $T_i,i = 1,\dots, d$ denote the $d$ independent workloads of the
$d$ randomly selected servers. We find:
\begin{align*}
\bar{F}_R(s)
&= \P\left\{E + \min_{i=1}^d T_i > s\right\}\\
&= e^{-s} + \int_0^{s} \bar{F}(s-t)^d e^{-t} dt.
\end{align*}
Due to \eqref{eq:expExact} and using standard integration techniques, this integral can be simplified to:
$$
\bar{F}_R(s) = e^{-s} \cdot \left(1 + \frac{1}{\lambda b^{1/(d-1)}} \cdot \int_{\left(\frac{b}{\lambda}\right)^{1/(d-1)}}^{e^s \left(\frac{b}{\lambda}\right)^{1/(d-1)} } (1 + x^{d-1})^{d/(1-d)} dx \right),
$$
where $b= \lambda^{1-d} - \lambda$ as before.
This is an integral that can be solved exactly to prove the statement. 
\end{proof}

\begin{remark}
It is easy to verify that the workload and response time distributions $F(s)$ and $F_R(s)$ have 
the same increasing failure rate $r(s) = f(s)/\bar{F}(s) = f_R(s)/\bar{F}_R(s)$.
\end{remark}

\begin{remark}
As $d \rightarrow \infty$, $F_R(s)$ tends to $e^{-s}$, as expected. 
\end{remark}

\begin{theorem}\label{eig:T_d}
The mean of the limiting response time distribution for the LL(d) policy with FCFS service
and exponential job sizes with mean $1$ is given by:
\begin{equation}\label{eq:T_kclosed}
T_d(\lambda)
=
\sum_{n=0}^{\infty} \frac{\lambda^{dn}}{1 + n \cdot (d-1)}.
\end{equation}
\end{theorem}
\begin{proof}
Let $E \sim \mbox{Exp}(1)$, we find:
\begin{align*}
T_d(\lambda) &= \E[E + \min\{T_1,\dots,T_d\}]\\
&=1 + \int_0^\infty \bar{F}(s)^dds.
\end{align*}
Using \eqref{eq:expExact} and standard integration techniques (mainly substitution), we can reduce this expression to:
$$
T_d(\lambda) = 1 + \frac{1}{\lambda^{d/(d-1)} \cdot (d-1)} \cdot \int_0^{\lambda / b} \frac{v^{1/(d-1)}}{(1+v)^{d/(d-1)}} dv.
$$
Using the substitution $y = \frac{v}{1 + v}$, one can show that the above integral reduces to
$$
1 + \frac{\lambda^d}{d} \cdot \prescript{}{2}F_1 \left( \frac{d}{d-1},1;1+\frac{d}{d-1};\lambda^d\right).
$$
As $\lambda^d \in (0,1)$, one can use \eqref{eq:sum2F1} and the claimed equality follows as $(1)_n=n!$ and $(d/(d-1))_n/(1+d/(d-1))_n = d/((n+1)(d-1)+1)$. 
\end{proof}

\begin{remark}
In the proof of Theorem \ref{eig:T_d} it is also possible to directly use \eqref{eq:FR} instead of relying on \eqref{eq:expExact}.
\end{remark}

\begin{remark}
Note that  $W_d(\lambda) = \lambda T_d(\lambda)$, which is expected due to Little's law and the fact that
the mean workload of a server under LL(d) service with exponential job sizes with mean $1$ is equal to the mean number
of jobs in such a server.
The relation $W_d(\lambda) = \lambda T_d(\lambda)$ also yields simple formulas for 
$T_2(\lambda)$ and $T_3(\lambda)$ due to Theorem \ref{th:Wd}. It is possible do derive similar expressions for larger $d$
values, but these become more and more complex as $d$ increases. 
\end{remark}

\begin{remark}
In \cite{gardnerOR} the mean of the limiting response time distribution
in case of exponential job sizes of the replication with cancellation-on-completion policy (under the {\it assumption} of the
independence ansatz) was argued to be equal to
$$
\E[T^{RR(d)}] = \frac{\prescript{}{2}F_1(1,1;1+\frac{d}{d-1};\frac{-\rho}{1-\rho})}{\mu d(1-\rho)}.
$$
This expression can be reduced to a simple sum formula as follows (using \eqref{eq:lintrans2} and \eqref{eq:sum2F1} as $\rho \in (0,1)$) 
\begin{align*}
\prescript{}{2}F_1(1,1;1+\frac{d}{d-1};\frac{-\rho}{1-\rho})
&= (1-\rho) \prescript{}{2}F_1(1,\frac{d}{d-1};1+\frac{d}{d-1};\rho)\\
&= (1-\rho) \sum_{n=0}^{\infty} \frac{(1)_n \left( \frac{d}{d-1}\right)_n}{\left(1+\frac{d}{d-1}\right)_n} \frac{\rho^n}{n!},
\end{align*}
which allows us to conclude that
\[\E[T^{RR(d)}] = \frac{1}{\mu} \sum_{n=0}^\infty \frac{\rho^n}{n(d-1)+d}.
\]
Note that $E[T^{RR(d)}]$ converges to $1/(d\mu)$ as $\rho$ tends to zero due to the independent execution times of the replicas in \cite{gardnerOR}.
\end{remark}

\section{Phase-type and deterministic job sizes}\label{sec:ph}

In Section \ref{sec:fixed} we proposed a fixed point iteration to compute the limiting workload distribution $F(s)$ under LL(d) 
for any service time distribution $G$, that was proven to converge if
$d \rho^d < 1$. We now show that $F(s)$ can also be directly obtained as the solution of a set of coupled ordinary
differential equations (ODEs) for any $\rho < 1$, provided that the job lengths follow a phase-type (PH) distribution. 
PH distributions are distributions with a modulating finite state background Markov chain \cite{latouche1} and any general positive-valued distribution 
can be approximated arbitrary closely with a PH distributions. Further, various fitting tools are available online 
for phase-type distributions (e.g., \cite{panchenko1,Kriege2014}). 
A PH distribution with $G(0)=0$ is fully characterized by a stochastic vector $\alpha = (\alpha_i)_{i=1}^n$ and a subgenerator matrix  $A = (a_{i,j})_{i,j=1}^n$
such that $\bar G(s) = \alpha e^{As} \textbf{1}$, where $\textbf{1}$ is a column vector of ones.
 
\begin{theorem}\label{eig:PHD}
Suppose the job lengths have a PH distribution characterized by $(\alpha,A)$, then the ccdf of the 
limiting workload distribution under the LL(d) policy satisfies:
\begin{align*}
\bar{F}'(s) & = -\lambda ((1-\bar{F}(s)^d) + \alpha A h(s)),\\
h'(s) & = (1-\bar{F}(s)^d) \textbf{1} + A h(s),
\end{align*}
with $\bar{F}(0) = \rho$, $h(0) = 0$ and
$h(s): \R \rightarrow \R^{n \times 1}$.
\end{theorem}
\begin{proof}
For $i \in \{1,\dots,n\}$ we define:
$$
h_{i}(s) = \int_0^s (1-\bar{F}(u)^d) e_i^T e^{(s-u)A} \textbf{1} du,
$$
where $e^T_i$ is the $i$-th row of the identity matrix $I_n$.
First note that $h_{i}(0) = 0$. We now derive a differential equation for $h_{i}(s)$. 
Using the equality $I_n = \sum_{k=1}^n e_k e_k^T$ we find :
\begin{align*}
h_{i}'(s)
&= (1-\bar{F}(s)^d) + \int_0^s (1 - \bar{F}(u)^d) e_i^T A I_n e^{(s-u)A} \textbf{1} du\\
&= (1-\bar{F}(s)^d) + \sum_{k=1}^n \int_0^s (1-\bar{F}(u)^d) e_i^T A e_k e_k^T e^{(s-u)A}\textbf{1} du \\
&= (1-\bar{F}(s)^d) + \sum_{k=1}^n a_{i,k} h_{k}(s).
\end{align*}
In matrix notation this yields:
$$
h'(s) = (1-\bar{F}(s)^d) \textbf{1} + A h(s).
$$
Due to \eqref{eq:F(s)} and $\bar G(s-u) = \alpha e^{(s-u)A} \textbf{1}$, 
we have $\bar{F}'(s) = -\lambda \alpha h'(s)$, which yields the equation for $\bar{F}'(s)$.
\end{proof}

We now generalize this result to the case where the service times 
are the sum of a deterministic random variable and a PH distribution.

\begin{theorem}\label{th:X+PH}
Assume the service times are the sum of a deterministic random variable with mean $\tau$ and
a phase-type distribution characterized by $(\alpha,A)$,
i.e., $\bar{G}(s) = I_{\{s\leq \tau\}} + I_{\{s > \tau\}} \alpha e^{(s-\tau) A} \textbf{1}$,
then the ccdf of the limiting workload distribution under the LL(d) policy satisfies:
\begin{align*}
\bar{F}'(s) &= \lambda(\bar{F}(s)^d - 1), & s \leq \tau,\\
\bar{F}'(s) &= -\lambda((1 - \bar{F}(s)^d) +\alpha A h(s-\tau)), & s > \tau,\\
h'(s) & = (1-\bar{F}(s)^d) \textbf{1} + A h(s),
\end{align*}
with $h(0) = 0$ and $\bar{F}(0) = \rho = \lambda (\tau + \alpha(-A)^{-1}\textbf{1})$.
\end{theorem}
\begin{proof}
We distinguish two cases: first let $s \in [0,\tau]$, we find that $\bar{F}(s) = \rho - \lambda \int_0^s 1 - \bar{F}(u)^d du$, deriving this equation once yields the first equation.

For the second note that we have (using the notation from the proof of Theorem \ref{eig:PHD}):
$$
\bar{F}(s) = \rho - \lambda \alpha h(s-\tau)  - \lambda \int_{s-\tau}^s (1-\bar{F}(u)^d) du.
$$
Taking the derivative and using the expression for $h'(s)$ found in Theorem \ref{eig:PHD} completes the proof.
\end{proof}

\begin{theorem}
If the job sizes are deterministic and equal to one, the ccdf $\bar{F}(s)$ is determined by $\bar{F}(0)=\lambda$, 
and
\begin{align*}
\bar{F}'(s) &= \lambda (\bar{F}(s)^d - 1) & s \in [0,1),\\
\bar{F}'(s) &= \lambda (\bar{F}(s)^d - \bar{F}(s-1)^d) & s \geq 1.
\end{align*}
\end{theorem}
\begin{proof}
The proof is similar to the proof of Theorem \ref{th:X+PH}.
\end{proof}

\begin{remark}
We note that the ODEs and DDEs presented in this section have a unique solution: the existence follows from the fact that \eqref{eq:F(s)} solves the ODE/DDE, 
while the uniqueness follows from \cite[Section 23, theorem A]{driver1977}.
\end{remark}

\begin{remark}
It is easy to compute the ccdf of the response time distribution $\bar{F}_R(s)$ given $\bar{F}(s)$ as
the probability that a new arrival joins a queue with a workload exceeding $s$ is given by $\bar{F}(s)^d$
under the LL(d) policy.
\end{remark}

\section{LL(d) versus SQ(d)}\label{sec:versus}

The aim of this section is to study the margin of improvement that can be achieved by using 
exact workload information as opposed to the coarser queue length information used
by SQ(d). This margin of improvement is of interest to understand the possible 
response time improvements offered by schedulers that implement late binding
(as discussed in the introduction). Furthermore, we also compare the SQ(d) policy with
the LL(d) policy where the job sizes of the latter take the late binding
overhead into account. We start by focusing on exponential job sizes, for which
we can also establish some closed form results.

\subsection{Exponential job sizes}
In this subsection we compare the limiting response time of the LL(d) and SQ(d) policies for exponential job sizes
with mean $1$ and FCFS service. This comparison provides an answer on the reduction in the response times that can be obtained
if the workloads at the different servers are known instead of the coarser queue length information.
To distinguish between the response times of both policies we make use of the superscripts $^{(LL)}$ and $^{(SQ)}$. 
For the SQ(d) policy the mean of the limiting response time distribution is given by \cite{mitzenmacher2}  
\[T_d^{(SQ)}(\lambda) = \frac{1}{\lambda} \sum_{k=1}^\infty \lambda^{\frac{d^k-1}{d-1}}.\]

\begin{theorem} \label{th:LLsmallerSQ}
The mean of the limiting response time distribution for the $LL(d)$ policy is smaller than the
mean for the $SQ(d)$ policy for exponential job sizes with mean $1$, moreover 
\begin{align*}
T_d^{(SQ)}(\lambda)- T_d^{(LL)}(\lambda) = \frac{1}{\lambda}\sum_{k=1}^{\infty} A_k, 
\end{align*}
where for $\lambda \in (0,1)$
\begin{align*}
A_k = \lambda^{\frac{d^{k+1}-1}{d-1}} - \sum_{n=1}^{d^k} \frac{\lambda^{nd + 1 + \frac{d^{k+1} - d^2}{d-1}}}{1 + n(d-1)+(d^k-d)} > 0.
\end{align*}
\end{theorem}
\begin{proof}
Due to \eqref{eq:T_kclosed}, we need to show:
$$
\sum_{n=1}^{\infty} \frac{\lambda^{dn+1}}{1 + n(d-1)}
\leq
\sum_{k=2}^{\infty} \lambda^{\frac{d^k-1}{d-1}}.
$$
To see this, we group the terms on the left hand side with $n \in \{\sum_{s=0}^{k-1} d^{s-1},\dots,\sum_{s=1}^k d^s\}$ together and compare their sum 
with the term $\lambda^{\frac{d^{k+1}-1}{d-1}}$ on the right hand side for $k \geq 1$.
We have
\begin{align*}
\sum_{n=1+\dots + d^{k-1}}^{d+\dots+d^k} \frac{\lambda^{nd+1}}{1+n(d-1)}
&<
\sum_{n=1+\dots + d^{k-1}}^{d+\dots+d^k} \frac{\lambda^{d (1+d+\dots+d^{k-1}) + 1}}{(1+\dots+d^{k-1})(d-1)}\\
&= \lambda ^{1 + d + \dots + d^k} = \lambda^{\frac{d^{k+1}-1}{d-1}}.
\end{align*}
Hence, the result follows. 

\end{proof}

\begin{theorem}\label{th:ratioSDdLLd}
For the ratio of the mean of the limiting response time distribution of SQ(d) and LL(d) for exponential job sizes with mean $1$ we have
$$
\lim_{\lambda \rightarrow 1} T_d^{(SQ)}(\lambda) / T_d^{(LL)}(\lambda) = \frac{d-1}{\log(d)}.
$$
\end{theorem}
\begin{proof}
Let $K \in \mathbb{N}$ be arbitrary and define:
$$
U_K(\lambda) = \frac{1+\sum_{k=1}^K \sum_{n=1+\dots+d^{k-1}}^{d+\dots+d^k}\frac{\lambda^{nd+1}}{1+n(d-1)}}{1+{\sum_{k=1}^K \lambda^{\frac{d^{k+1}-1}{d-1}}}}.
$$
We note that we have:
$$
\lim_{\lambda \rightarrow 1} \lim_{K\rightarrow \infty} U_K(\lambda) = \lim_{\lambda \rightarrow 1}\frac{T_d^{(LL)}(\lambda)}{T_d^{(SQ)}(\lambda)}.
$$
On the other hand (with $\psi$ the Digamma function \cite[Chapter 6]{abramowitz64}) we have:
\begin{align*}
\lim_{K\rightarrow \infty} \lim_{\lambda \rightarrow 1} U_K(\lambda)
&=\lim_{K\rightarrow \infty} \frac{\sum_{k=0}^K \sum_{n=\frac{d^k-1}{d-1}}^{d\frac{d^k-1}{d-1}} \frac{1}{1+n(d-1)}}{\sum_{k=0}^K 1}\\
&= \frac{1}{d-1}  \lim_{K\rightarrow \infty} \frac{ \sum_{k=0}^K\psi \left(\frac{d^{k+1}}{d-1}\right)-\psi
   \left(\frac{d^k}{d-1}\right)}{\sum_{k=0}^K 1}.\\
\end{align*}
Since $\lim_{k\rightarrow \infty} \psi \left(\frac{d^{k+1}}{d-1}\right)-\psi
   \left(\frac{d^k}{d-1}\right) = \log(d)$, we may apply the Stolz-Cesaro theorem to assert that
$$
\lim_{K\rightarrow\infty} \lim_{\lambda \rightarrow 1} U_K(\lambda) = \frac{\log(d)}{d-1}.
$$
If we may interchange the limits this would incur:
$$
\lim_{\lambda \rightarrow 1} \frac{T_d^{(LL)}(\lambda)}{T_d^{(SQ)}(\lambda)} = \frac{\log(d)}{d-1}.
$$
An application of the Moore-Osgood theorem \cite[p100]{graves1946} implies that we may indeed interchange limits: as $U_K$ and $U = \lim_{K \rightarrow \infty} U_K$ 
are continuous functions defined on the compact set $[0,1]$ and $U_K$ converges pointwise to $U$, it follows that this convergence is also uniform. Moreover, we trivially have pointwise convergence of $\lim_{\lambda \rightarrow 1} U_K(\lambda)$.
\end{proof}

\begin{remark}
As $(d-1)/\log(d)$ tends to infinity as $d$ becomes large, we note that for any $c>0$ there exists a $\lambda$ and $d$ such
that the ratio $T_d^{(SQ)}(\lambda) / T_d^{(LL)}(\lambda) > c$. In other words, 
for arbitrary $\lambda$ and $d$, there is no bound on how much worse 
the SQ(d) policy performs than the LL(d) policy. 
\end{remark}

In Figure \ref{fig:SQvsLLmean} we plot the ratio $T_d^{(SQ)}(\lambda)/T_d^{(LL)}(\lambda)$ as a function
of $\lambda$. We note that this ratio  increases with $\lambda$ and approaches a constant as
$\lambda$ approaches one. Looking at this figure, the limit values for the ratio $T_d^{(SQ)}(\lambda)/T_d^{(LL)}(\lambda)$ 
as $\lambda$ tends to one may appear to be less than $(d-1)/\log(d)$ (as shown in Theorem \ref{th:ratioSDdLLd}), 
but this is simply due to the fact
that this ratio still increases significantly between $0.999$ and $1$.
From this figure we may conclude that the increase in the mean of the limiting response time distribution
by using the coarser queue length information instead of the exact workload is below $50\%$ when $d=2$
for exponential job sizes. For larger $d$ we see a more significant increase under high load.

We further note that the curves for different $d$ values cross one another. 
Intuitively this can be understood
by noting that for $\lambda$ small many jobs select an idle server and when an idle server is selected
knowing the queue length is equally good as knowing the workload. When $d$ increases it becomes more likely that an idle server is
selected and thus we expect the mean response time ratio to decrease with increasing $d$ when $\lambda$ is small. 
For large $\lambda$ it becomes unlikely that one of the selected queues is idle and SQ(d) has to rely
on the coarser queue length information. When $\lambda$ is large, we therefore
see a larger loss of more information as $d$ increases and thus the  mean response time ratio now increases 
with increasing $d$.

\begin{figure}[t]
\center
\includegraphics[width=0.95\linewidth]{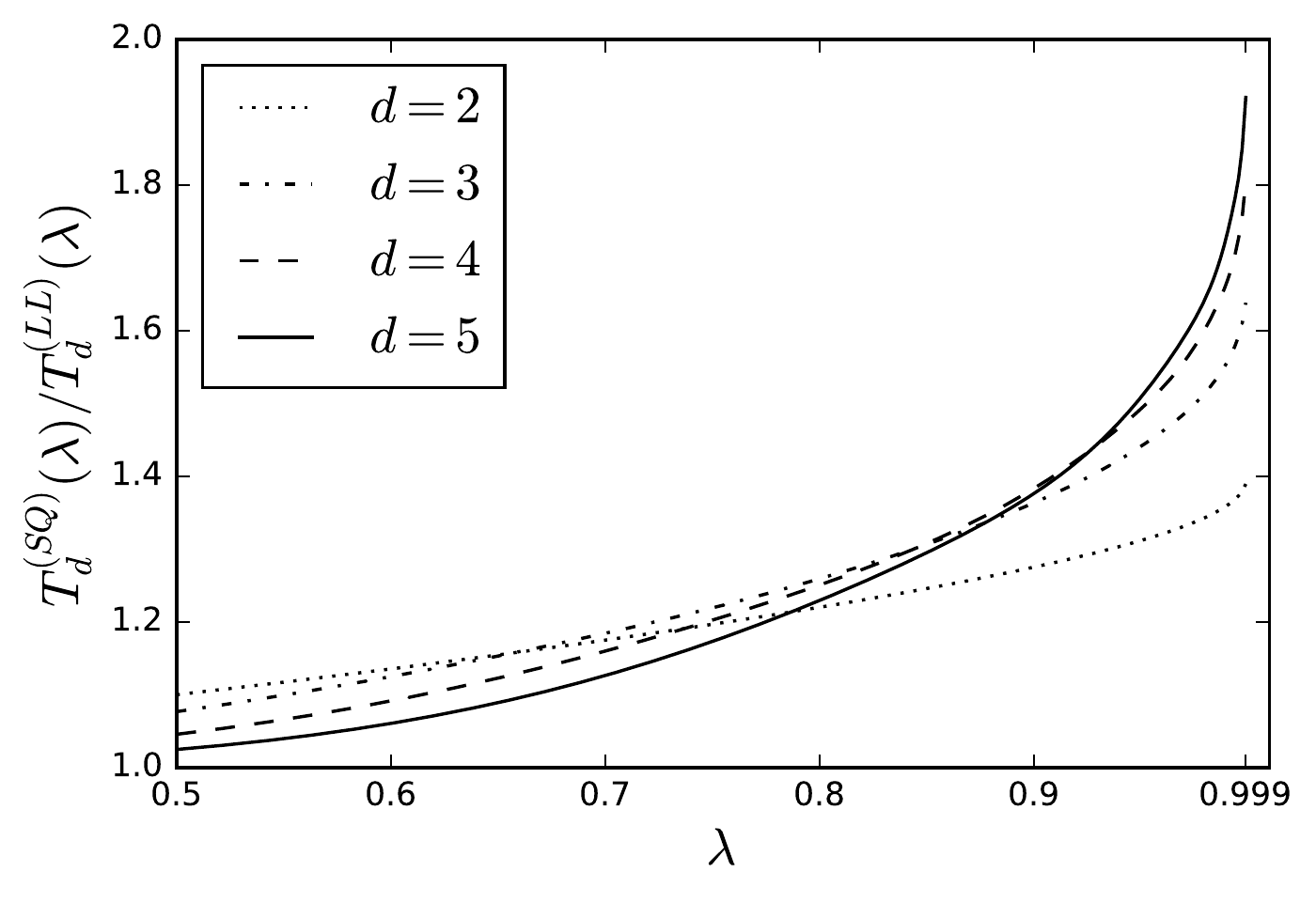}
\caption{Ratio of the mean of the limiting response time distribution of SQ(d) and LL(d) for exponential job
sizes with mean $1$, FCFS service as a function of $\lambda$.}
\label{fig:SQvsLLmean}
\end{figure}

Apart from comparing the mean response times, we can also easily compare the response time distribution
of the LL(d) and SQ(d) policy. For the SQ(d) policy it is not hard to establish that the ccdf of
the limiting response time distribution can be written as 
\begin{align}\label{eq:FRSQd}
\bar{F}_R^{(SQ)}(s) &= \sum_{k=1}^\infty \left(\lambda^{(d^{k-1}-1)d/(d-1)}-\lambda^{(d^{k}-1)d/(d-1)}\right) \sum_{n=0}^{k-1} \frac{s^n}{n!} e^{-s}
\nonumber \\
&=\sum_{n=0}^\infty  \frac{s^n}{n!} e^{-s} \lambda^{(d^n-1)d/(d-1)},
\end{align}
by noting that a job that joins a queue of length $k-1$ has an Erlang-$k$ distributed response time
for exponential job sizes. 
Figure \ref{fig:SQvsLLdist95} depicts the response
time distributions for $\lambda = 0.95$ and $d=2,3$ and $4$. We note that $\bar{F}_R(s)$
decreases as a function of $d$ and $\bar{F}^{(SQ)}_R(s)$ dominates
$\bar{F}^{(LL)}_R(s)$ for all $s > 0$. The next theorem proves an even stronger result.

\begin{figure}[t]
\center
\includegraphics[width=0.95\linewidth]{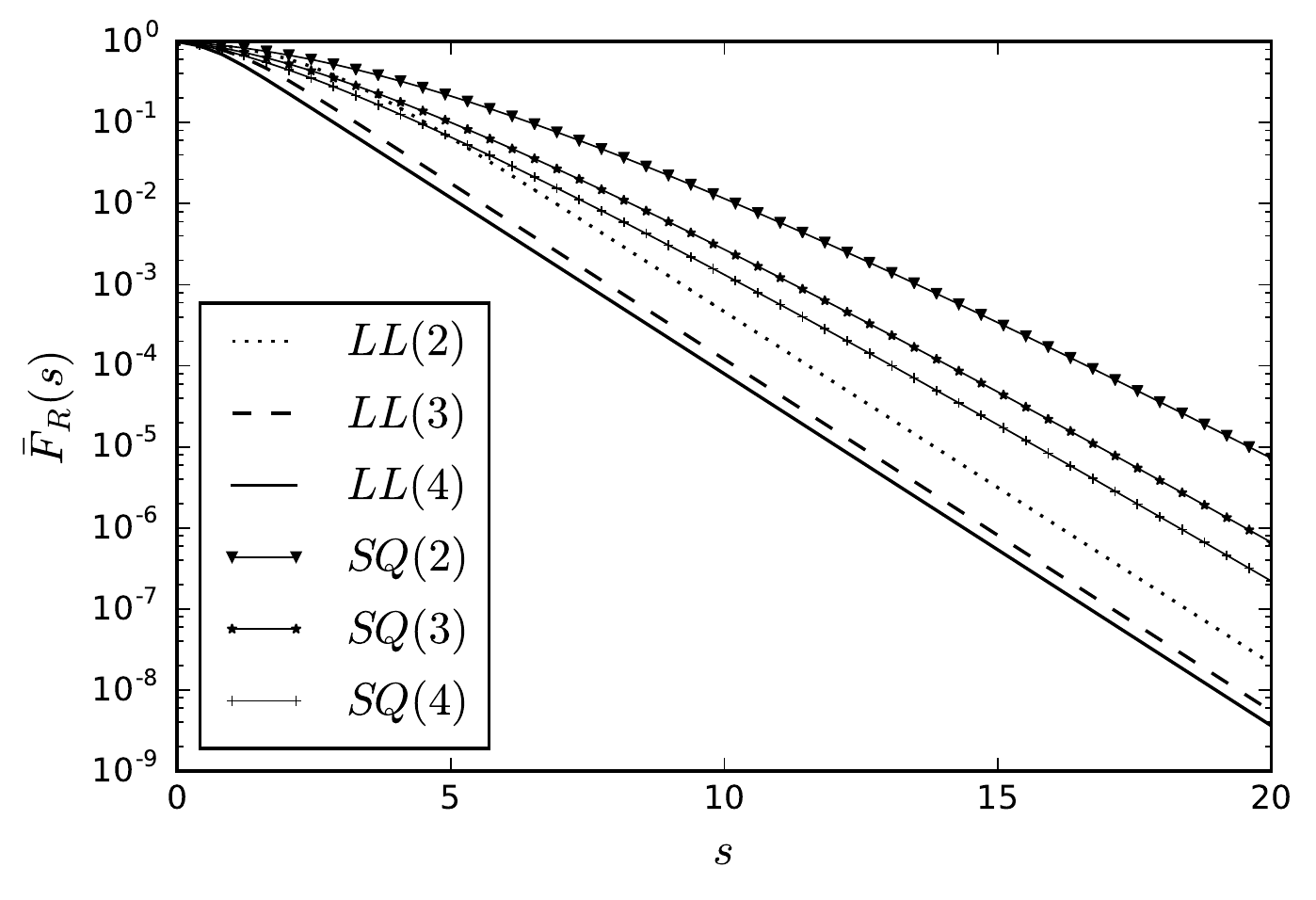}
\caption{Limiting response time distribution of SQ(d) and LL(d) for exponential job
sizes with mean $1$, FCFS service and $\lambda = 0.95$.}
\label{fig:SQvsLLdist95}
\end{figure}

\begin{theorem}
The function $f(s) = \bar{F}_R^{(SQ)}(s)/\bar{F}_R^{(LL)}(s)$ is non-decreasing on $[0,\infty)$, thus
$\bar{F}_R^{(SQ)}(s) \geq \bar{F}_R^{(LL)}(s)$ for all $s$.
\end{theorem}
\begin{proof}
It suffices to show that $f'(s) \geq 0$ for $s>0$ (as $\bar{F}_R^{(SQ)}(0) =\bar{F}_R^{(LL)}(0) = 1$).
Denote $\mu = \lambda^d$. Using \eqref{eq:FRSQd} and \eqref{eq:FR}, the condition $f'(s) \geq 0$ can be restated as
\begin{align*}
\frac{\sum_{k=0}^\infty \mu^{\frac{d^k-1}{d-1}} (\mu^{d^k}-1)\frac{s^k}{k!}}{\sum_{k=0}^\infty \mu^{\frac{d^k-1}{d-1}}\frac{s^k}{k!}} + \frac{(1-\mu) e^{(d-1) s}}{\mu + (1-\mu) e^{(d-1)s}} \geq 0.
\end{align*}
By rearranging terms this is equivalent to showing:
$$
e^{(d-1)s} \left( \sum_{k=0}^\infty \mu^{d^k} \mu^{{\frac{d^k-1}{d-1}}} \frac{s^k}{k!} \right) \geq \frac{\mu}{1-\mu} \sum_{k=0}^\infty \mu^{\frac{d^k-1}{d-1}} (1-\mu^{d^k}) \frac{s^k}{k!}.
$$
For the left hand side we find, by using the Taylor expansion of $e^{(d-1)s}$ and applying Merten's theorem:
\begin{align*}
e^{(d-1)s} \left( \sum_{k=0}^\infty \mu^{d^k} \mu^{{\frac{d^k-1}{d-1}}} \frac{s^k}{k!} \right)
&= \sum_{n=0}^\infty \frac{s^n}{n!} \sum_{k=0}^n \binom{n}{k} (d-1)^{n-k} \mu^{d^k} \mu^{\frac{d^k-1}{d-1}}.
\end{align*}
It therefore suffices to show that the inequality holds for all coefficients of $\frac{s^n}{n!}$, i.e.~it remains to show that:
$$
\frac{\mu}{1-\mu} \mu^{\frac{d^n-1}{d-1}} (1-\mu^{d^n}) \leq \sum_{k=0}^n \binom{n}{k} (d-1)^{n-k} \mu^{d^k} \mu^{\frac{d^k-1}{d-1}}.
$$
By noting that $\frac{1-\mu^{d^n}}{1-\mu} \leq d^n$, the result follows if the following holds
$$
d^n \leq \sum_{k=0}^n \binom{n}{k} (d-1)^{n-k} \mu^{\frac{d^{k+1}-1}{d-1} - \frac{d^{n}-1}{d-1} - 1},
$$
We clearly have an equality in $\mu = 1$ (and for $n=0$). It therefore suffices to show that the right hand side decreases for $\mu \in [0,1]$
for $n > 0$. The first $n$ terms are all 
convex decreasing, while the last term is convex increasing. The derivative of the sum of the first and last term in $\mu = 1$ is $(d^n-1)(1-(d-1)^{n-1}) \leq 0$.
Since the derivative of a convex function on $[0,1]$ is maximized in $1$, the sum of the first and last term is decreasing and 
we may conclude that $f'(s) \geq 0$.
\end{proof}
\subsection{Impact of job variability}\label{sec:SCV}
In this subsection we study the impact of the job size variability on the ratio of the SQ(d) and LL(d) 
mean of the limiting response time distribution. In real systems a significant part of the total workload is often offered by a small
fraction of long jobs, while the remaining workload consists mostly of (very) short jobs \cite{Sparrow}. 
For simplicity we represent these workloads as a hyperexponential (HEXP) distribution (with $2$ phases) 
such that we can vary the job size variability in a systematic manner.  
More precisely, with probability $p$ a job is a type-$1$ job and has an exponential length with parameter $\mu_1 > 1$ 
and with the remaining probability $1-p$ a job is a type-$2$ job and
has exponential length with parameter $\mu_2 < 1$. Hence, the type-$2$ jobs are longer on average
and we therefore sometimes refer to the type-$2$ jobs as the {\it long} jobs. 
The parameters $p, \mu_1$ and $\mu_2$ are set such that the following three values are matched: 
(i) mean job length (set to one), (ii) the squared coefficient of variation (SCV) and (iii) a shape parameter $f$,
using the following equations:
\begin{align*}\label{eq:HyperExp}
 \mu_{1} &= \frac{SCV+(4f-1)+\sqrt{(SCV-1)(SCV-1+8f\bar f)}}{2f(SCV+1)},\\
 \mu_{2} &= \frac{SCV+(4\bar f-1)-\sqrt{(SCV-1)(SCV-1+8f\bar f)}}{2\bar f(SCV+1)}, 
\end{align*}
with $\bar f = 1-f$  and $p = \mu_1 f$. 
The shape parameter $f \in (0,1)$ represents the fraction of the workload that is offered by the type-$1$ jobs.

\begin{figure}[t]
\center
\includegraphics[width=0.9\linewidth]{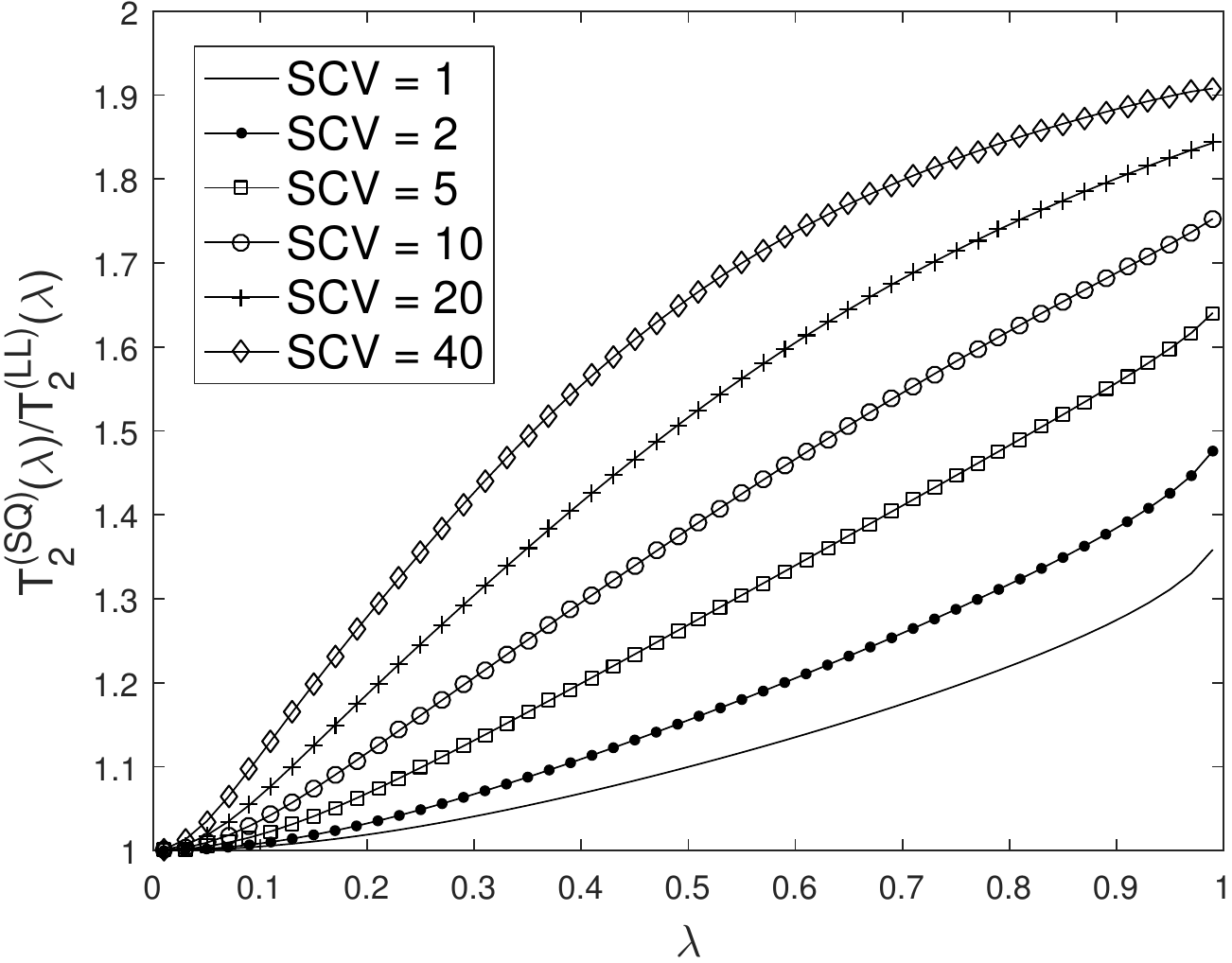}
\caption{Ratio of the mean of the limiting response time distribution of SQ(2) and LL(2) for hyperexponential
job sizes with mean $1$, shape parameter $f=1/2$ and FCFS service as a function of $\lambda$.}
\label{fig:LLDvsSQDhighSCVd2}
\end{figure}

The mean of the limiting response time distribution for the LL(d) policy can be computed in a fraction of a second for any $\rho < 1$
by making use of Theorem \ref{eig:PHD}. For the SQ(d) policy we use a fixed point 
iteration to determine the stationary queue length distribution of the cavity process associated to the
equilibrium environment \cite{bramsonLB}. More specifically, we determine the queue length distribution 
of a sequence of M/G/1 FCFS queues with a queue length dependent arrival rate $\lambda$,
where the queue length distribution determined during the $n$-th iteration determines the
arrival rates of the $n+1$-th iteration, until the queue length distribution converges (starting from the empty distribution). 
While the queue length distribution of such a queue can be computed in a very fast manner when the job
sizes follow a phase-type distribution (or are deterministic), the number of iterations
needed increases sharply as $\rho$ approaches $1$. 
This prevents us from studying what happens in the limit as $\rho$ tends to one.

Figure \ref{fig:LLDvsSQDhighSCVd2} depicts the ratio of the mean of the limiting response time distribution
of the SQ(d) and LL(d) policies when $d=2$ and $f=1/2$ (meaning half of the workload is offered by the {\it long} jobs). 
This ratio increases when the jobs sizes become more variable, which is expected
as having precise workload information should be more valuable when jobs vary significantly in size. 
The results indicate that a mechanism like late binding can offer substantial gains even at fairly low loads
if the job sizes vary significantly (and the round-trip time to fetch the job can be neglected).
The results for $f=1/10$, which implies that $90\%$ of the workload is offered by the {\it long} jobs, 
are very similar (and therefore not depicted). For $d>2$ these ratios tend to increase under sufficiently high loads 
as in the exponential case.

\begin{figure}[t]
\center
\includegraphics[width=0.9\linewidth]{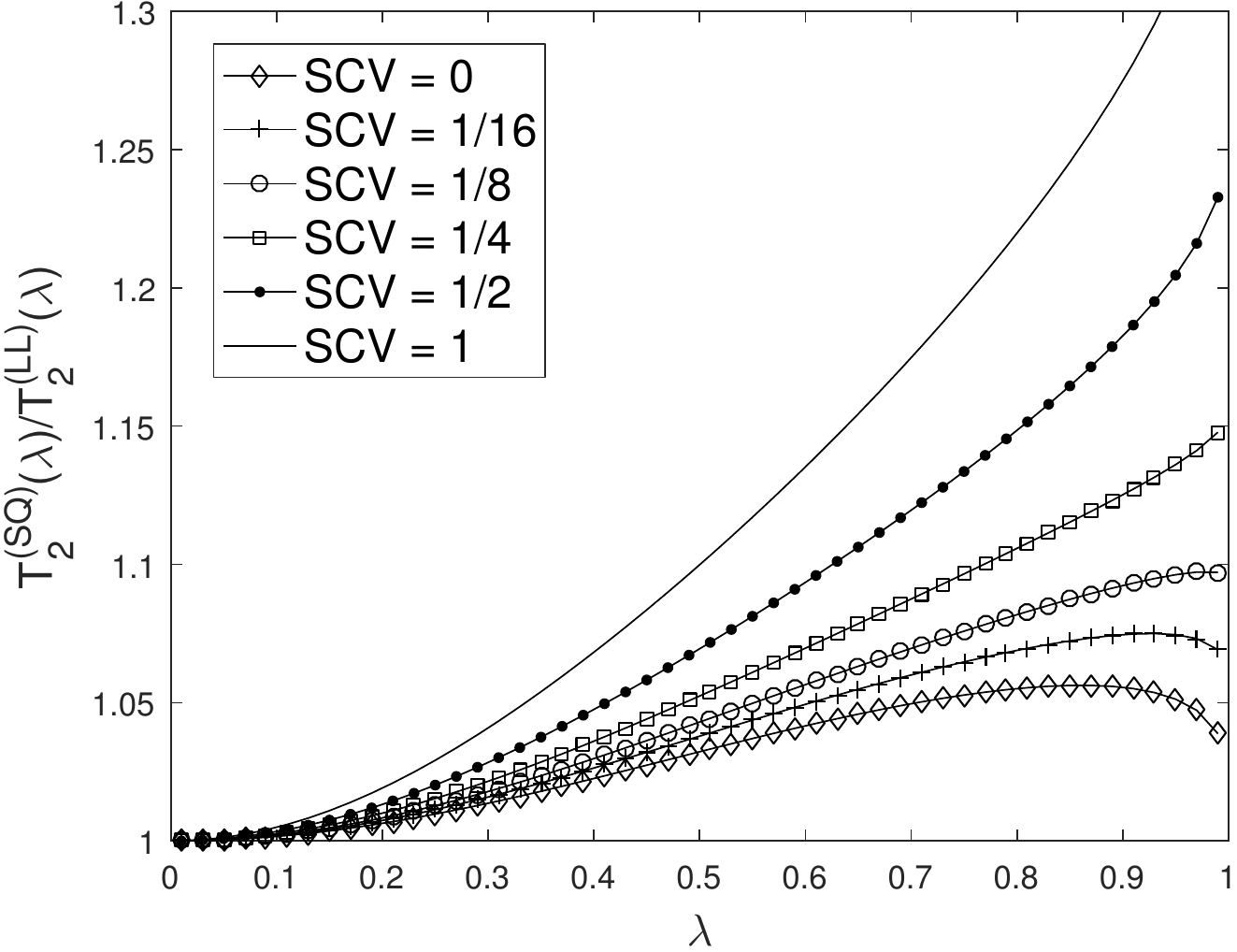}
\caption{Ratio of the mean of the limiting response time distribution of SQ(2) and LL(2) for hyperexponential
job sizes with mean $1$, shape parameter $f=1/2$ and FCFS service as a function of $\lambda$.}
\label{fig:LLDvsSQDlowSCVd2}
\end{figure}

For completeness we also present some results for job sizes with an SCV below $1$ in Figure \ref{fig:LLDvsSQDlowSCVd2}. 
In this case we cannot make use of a hyperexponential distribution and therefore consider Erlang-$k$ distributed 
and deterministic job sizes instead. 
This figure shows that as $\lambda$ approaches $1$ the ratio of the means of the limiting response time distribution
starts to decrease for sufficiently small SCVs. In fact, studying this ratio for $\lambda$ values closer to $1$ as depicted in   Figure \ref{fig:LLDvsSQDlowSCVd2}
suggests that this ratio decreases to $1$ for deterministic job sizes. This seems to make sense intuitively as
for $\lambda$ approaching one, the queue lengths become long and knowing the coarser queue length information
is almost as good as knowing the exact workload.

\subsection{Late binding overhead}\label{sec:overhead}
In the previous subsection we shed light on the margin of improvement that late binding can provide
compared to the classic SQ(d) policy assuming that the jobs can be fetched from the dispatchers in negligible time.
In this section we take the idleness caused by late binding into account.
We do this by comparing the mean of the limiting response time distribution of the SQ(d) policy with the mean of the LL(d) policy, 
where the size of each job under the LL(d) policy is incremented by a deterministic quantity $\tau$ that represents the overhead, that is,
the time that the server remains idle under late binding while fetching the job. We denote the mean of the limiting response time in the latter case 
as $T_{d,\tau}^{(LL)}$ and rely on Theorem \ref{th:X+PH} for its computation. We consider the same job size distributions 
(with average job size equal to one) as in the previous section.

\begin{figure}[t]
\center
\includegraphics[width=0.9\linewidth]{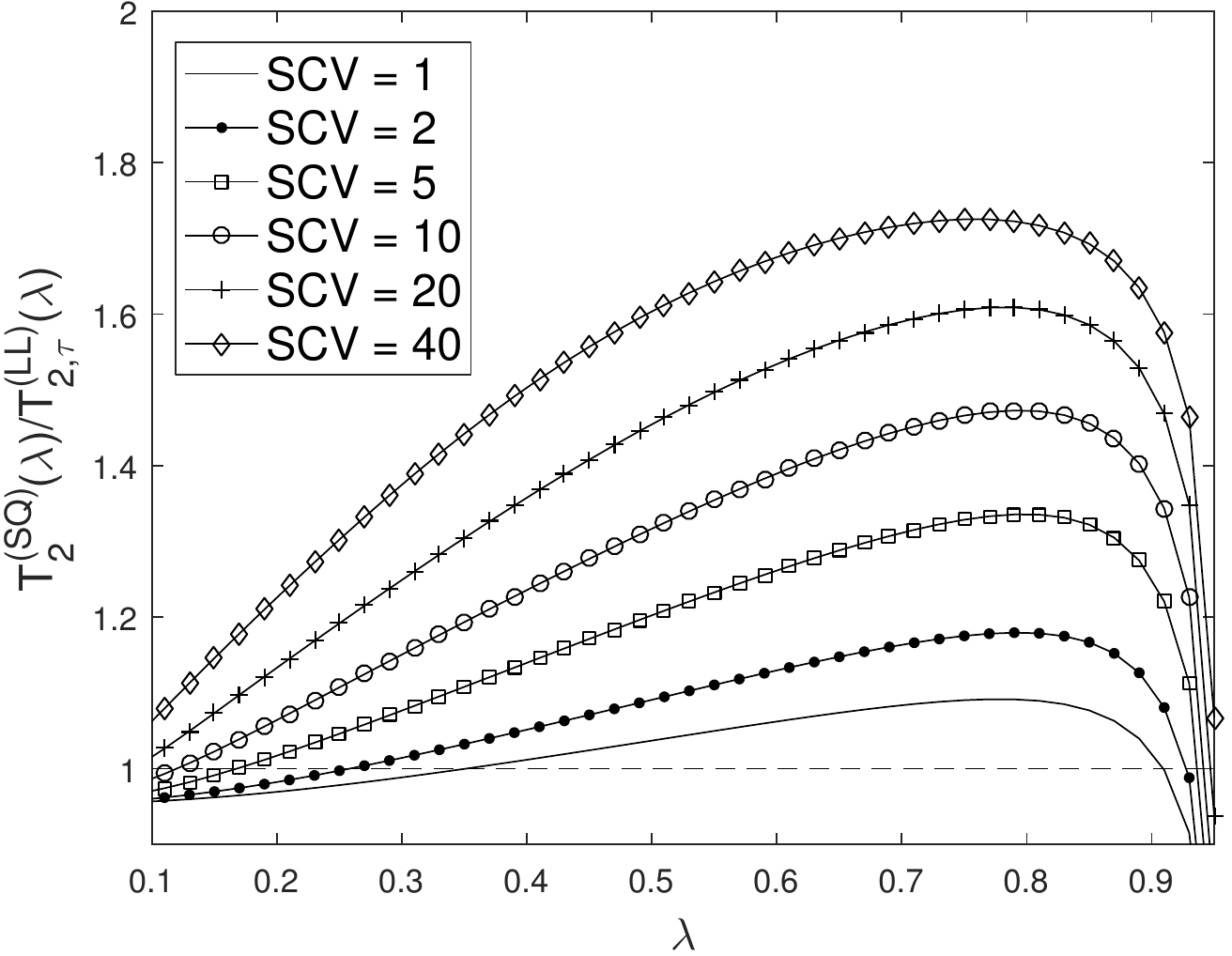}
\caption{Ratio of the mean of the limiting response time distribution of SQ(2) and LL(2)
with $5\%$ overhead (i.e., $\tau = 0.05$) for hyperexponential
job sizes with mean $1$, shape parameter $f=1/2$ and FCFS service as a function of $\lambda$.}
\label{fig:LLDvsSQDhighSCVd2tau005}
\end{figure}

In Figure \ref{fig:LLDvsSQDhighSCVd2tau005} the ratio $T_d^{(SQ)}/T_{d,\tau}^{(LL)}$ is shown
as a function of $\lambda$ for the case where $\tau = 0.05$, meaning each job induces an
idle server period with a length equal to $5\%$ of the mean job size.  It indicates that for a very wide range of
arrival rates $\lambda$, late binding offers substantial gains over the SQ(d) policy even with an overhead of 
$5\%$. For systems with high job size variability, this range even includes arrival rates above $0.9$. Note that the overhead of 
the scheduler implementation in \cite{Sparrow} was estimated to be below $2\%$. 

In fact for medium loads much higher amounts of overhead can be tolerated by the LL(d) policy before it becomes inferior to SQ(d). This
is illustrated in Figure \ref{fig:LLDvsSQD_tau_SCV20}, where we plot the largest $\tau$ value for which $T_{d,\tau}^{(LL)} \leq T_d^{(SQ)}$
when the SCV was set to $20$.
We observe that overheads of $25\%$ and more can be tolerated for system workloads around $50\%$.   

\begin{figure}[t]
\center
\includegraphics[width=0.9\linewidth]{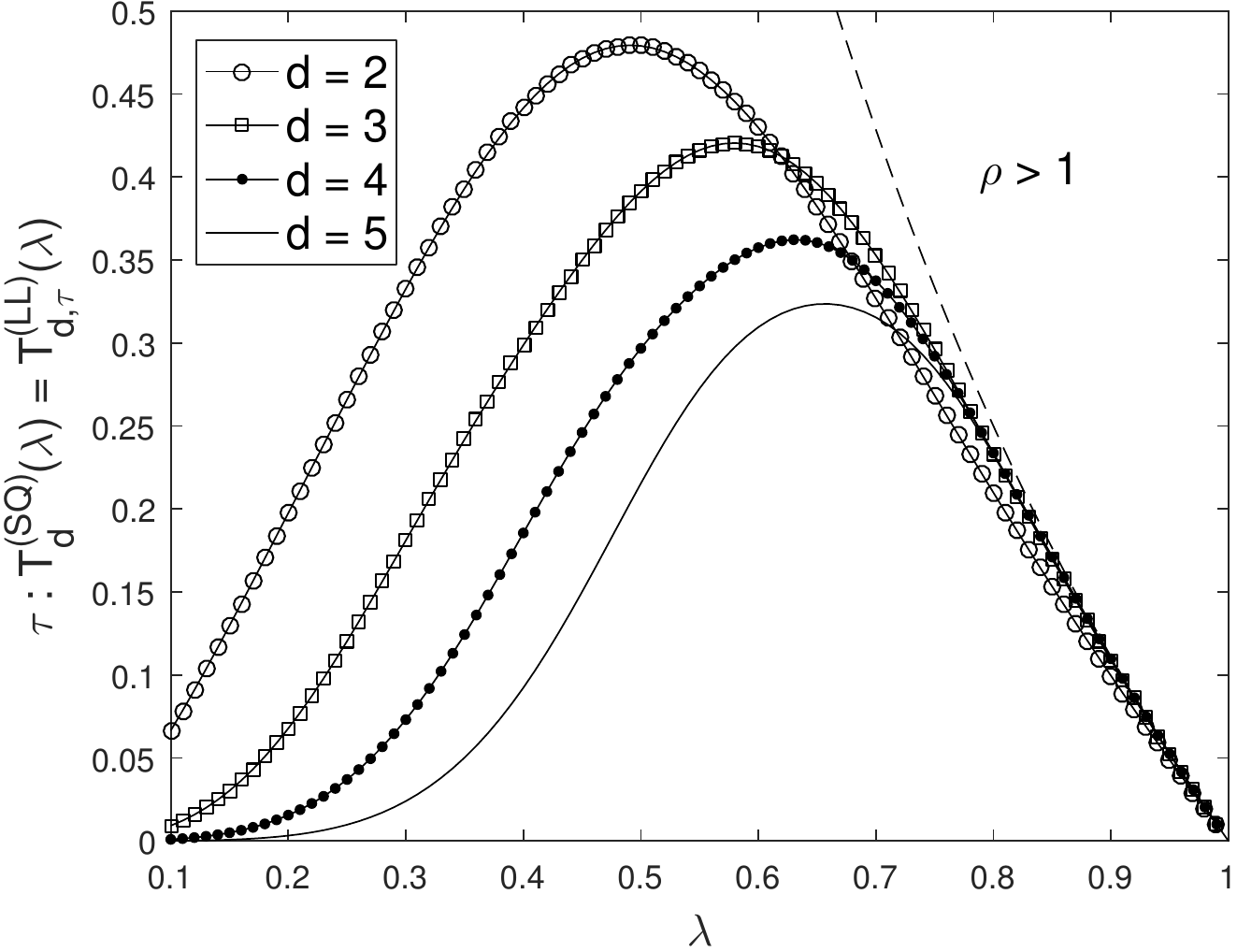}
\caption{Degree of delay that the LL(d) policy can tolerate without being outperformed by SQ(d)
as a function of $\lambda$ for hyperexponential job sizes with mean $1$, SCV = 20 and $f=1/2$, i.e., 
the largest $\tau$ such that $T_{d,\tau}^{(LL)} \leq T_d^{(SQ)}$.}
\label{fig:LLDvsSQD_tau_SCV20}
\end{figure}

\begin{figure*}[t]
\begin{subfigure}{.45\textwidth}
\centering
\includegraphics[width=0.9\textwidth]{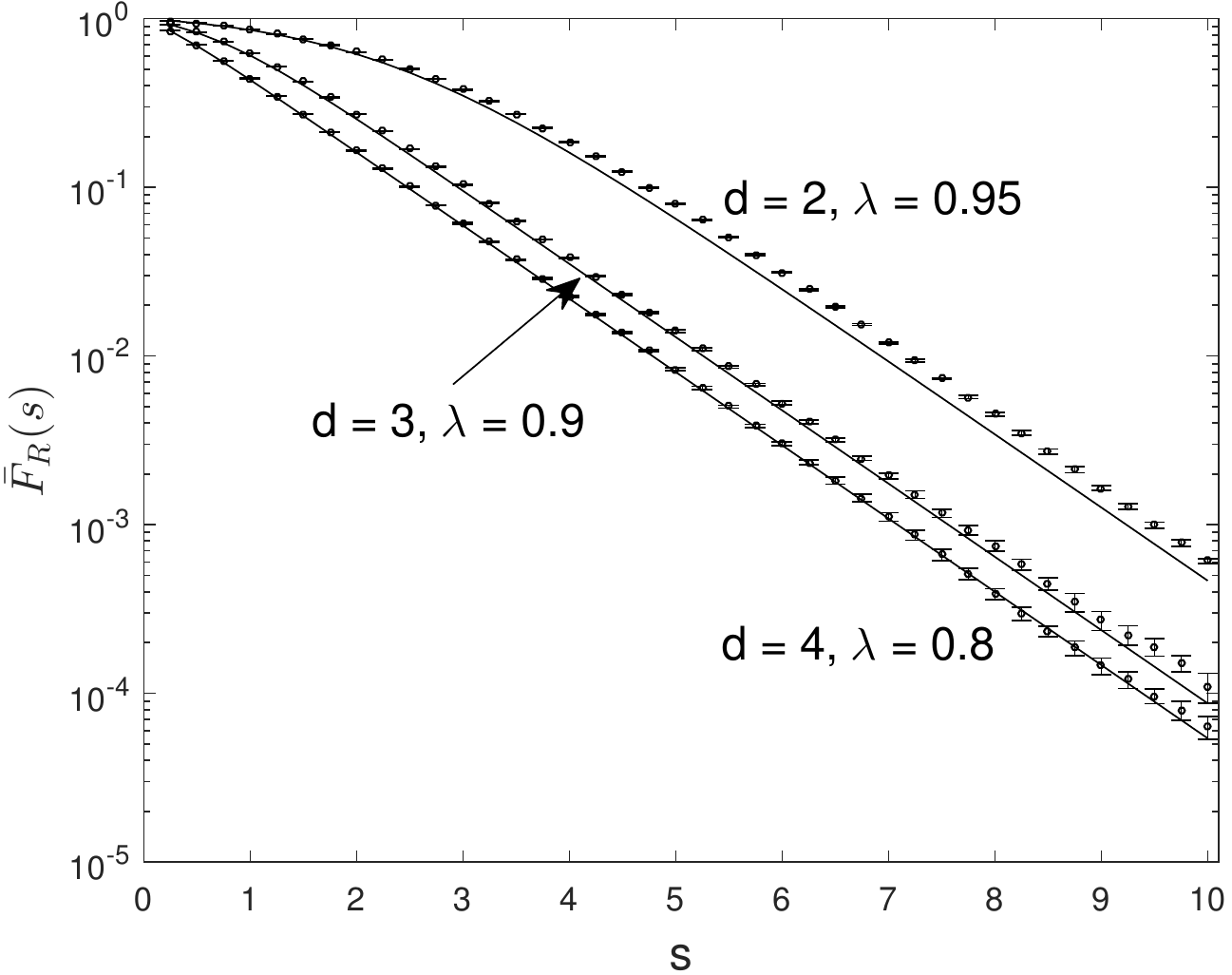}
\caption{Exponential jobs, $N = 100$}
\label{fig:validateEXP}
\end{subfigure}
\begin{subfigure}{.45\textwidth}
\centering
\includegraphics[width=0.9\textwidth]{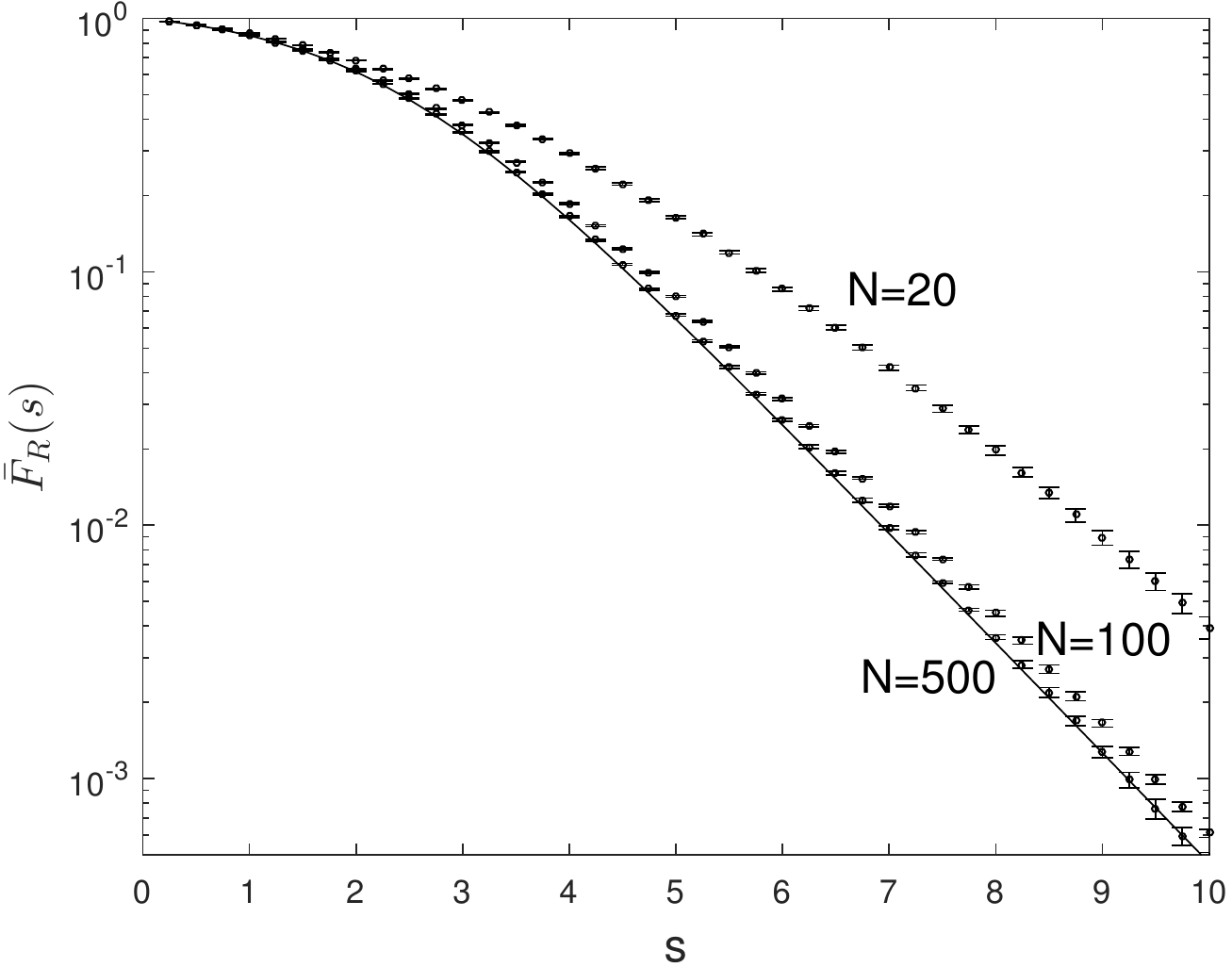}
\caption{Exponential jobs, $d = 2$, $\lambda = 0.95$}
\label{fig:validateEXP2}
\end{subfigure}
\vspace*{4mm}
\begin{subfigure}{.45\textwidth}
\centering
\includegraphics[width=0.9\textwidth]{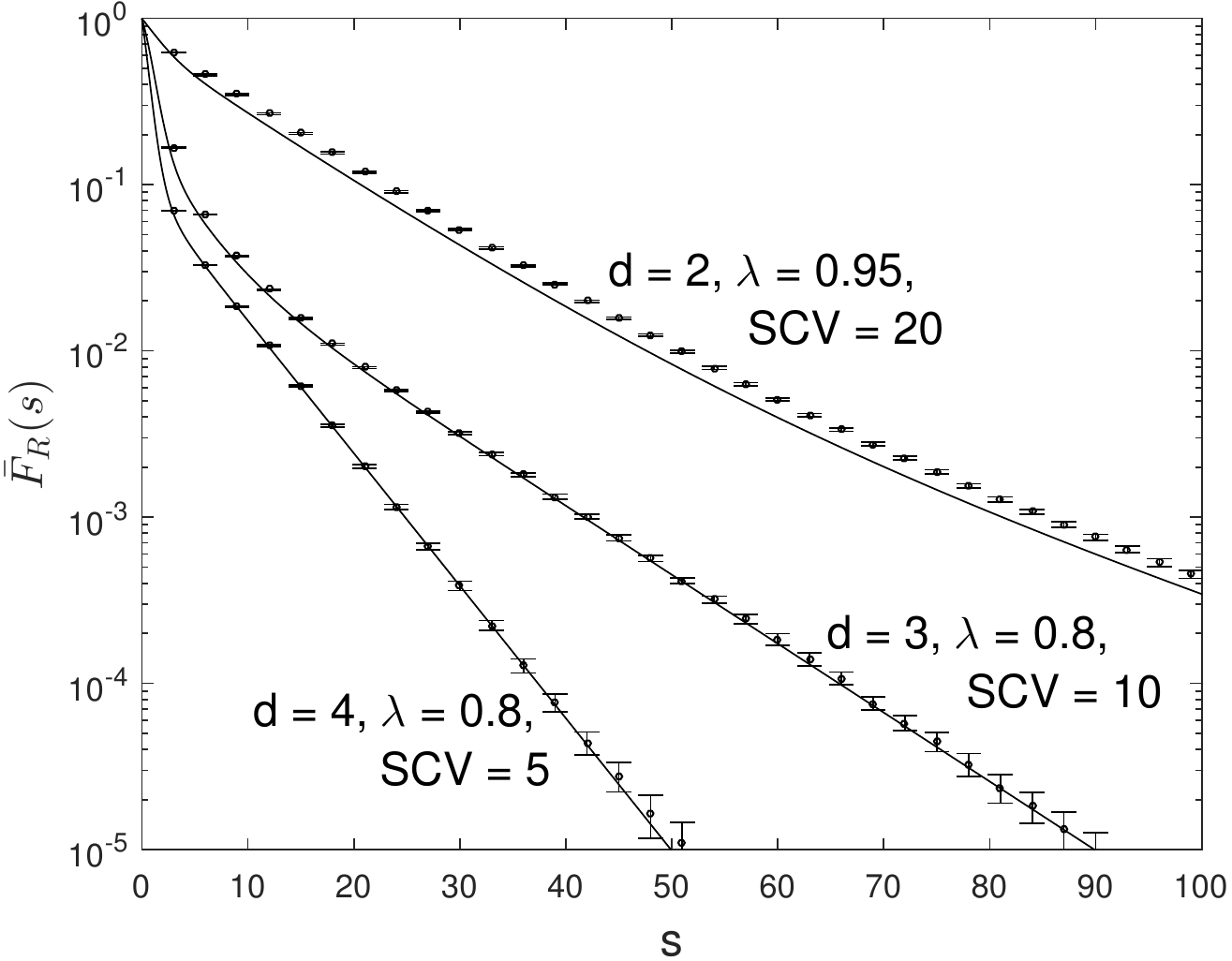}
\caption{Hyperexponential jobs, $N = 100$}
\label{fig:validateHEXP}
\end{subfigure}
\begin{subfigure}{.45\textwidth}
\centering
\includegraphics[width=0.9\textwidth]{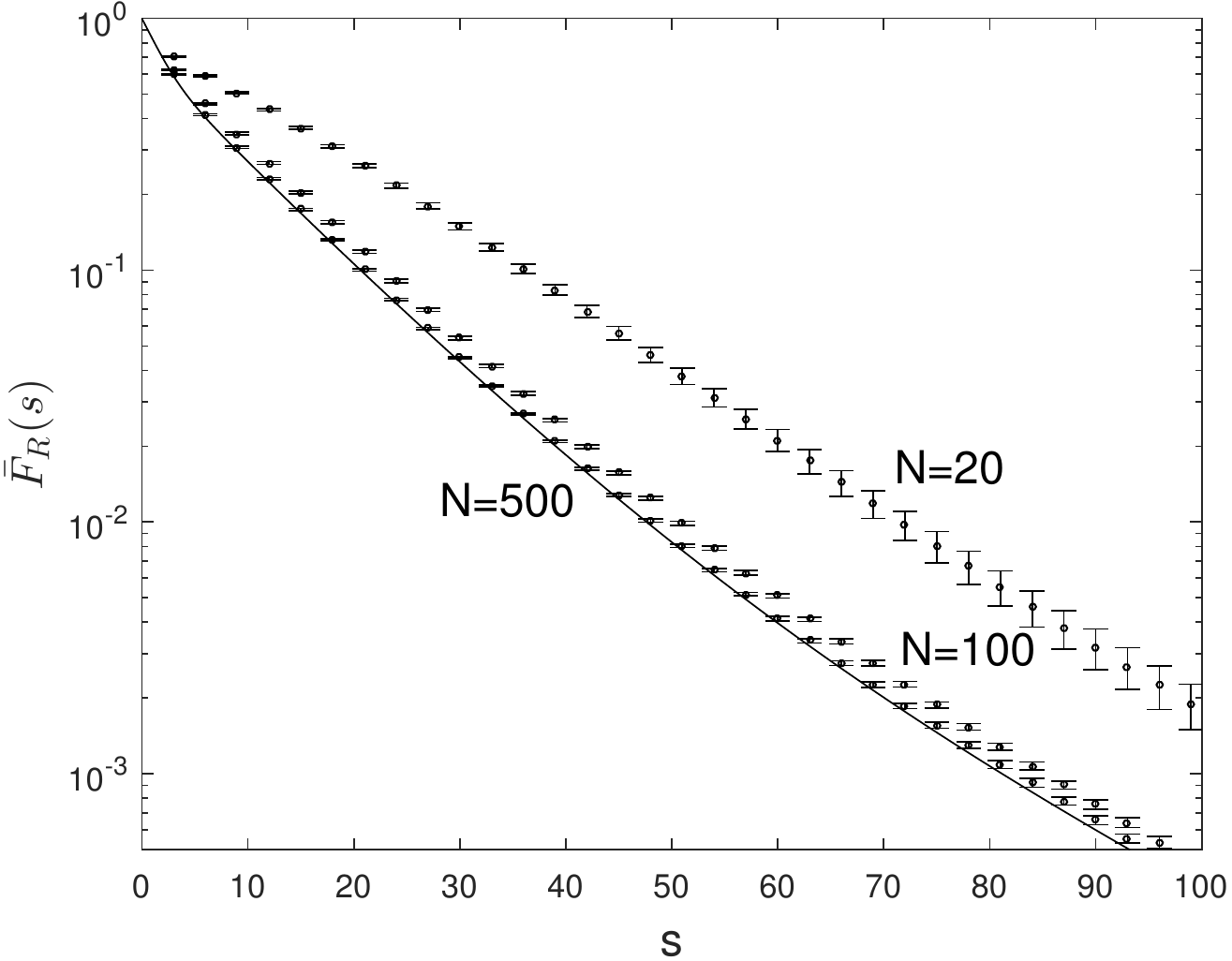}
\caption{Hyperexponential jobs, $d = 2$, $\lambda = 0.95$, $SCV = 20$}
\label{fig:validateHEXP2}
\end{subfigure}
\caption{Limiting response time distribution vs.~simulation 
for $N$ servers with (hyper)exponential job sizes with mean $1$.
The full line represents the limiting response time distribution.}
\end{figure*}

\section{Finite system accuracy}\label{sec:finite}

In this section we briefly compare the limiting response time distribution with 
simulation experiments where the number of servers $N$ is finite. All simulation runs simulate the system
up to time $t=10^7/N$ and use a warm-up period of $30\%$. 

\begin{figure*}[t]
\begin{subfigure}{.45\textwidth}
\centering
\includegraphics[width=0.9\textwidth]{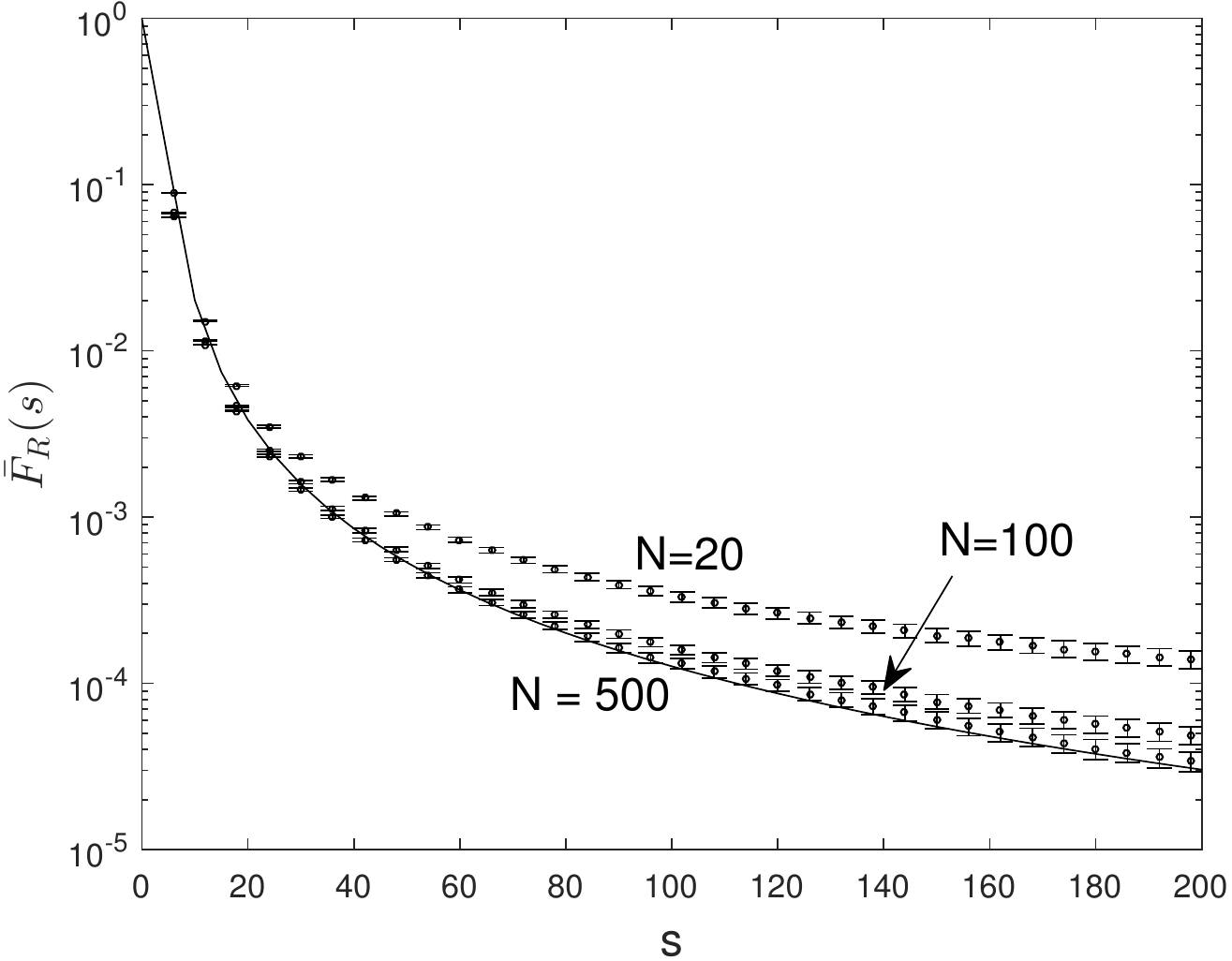}
\caption{Power law job sizes with $\bar{G}(s) = x^{-2}$, $d=2, \rho = 0.8$}
\label{fig:validatePOW}
\end{subfigure}
\begin{subfigure}{.45\textwidth}
\centering
\includegraphics[width=0.9\textwidth]{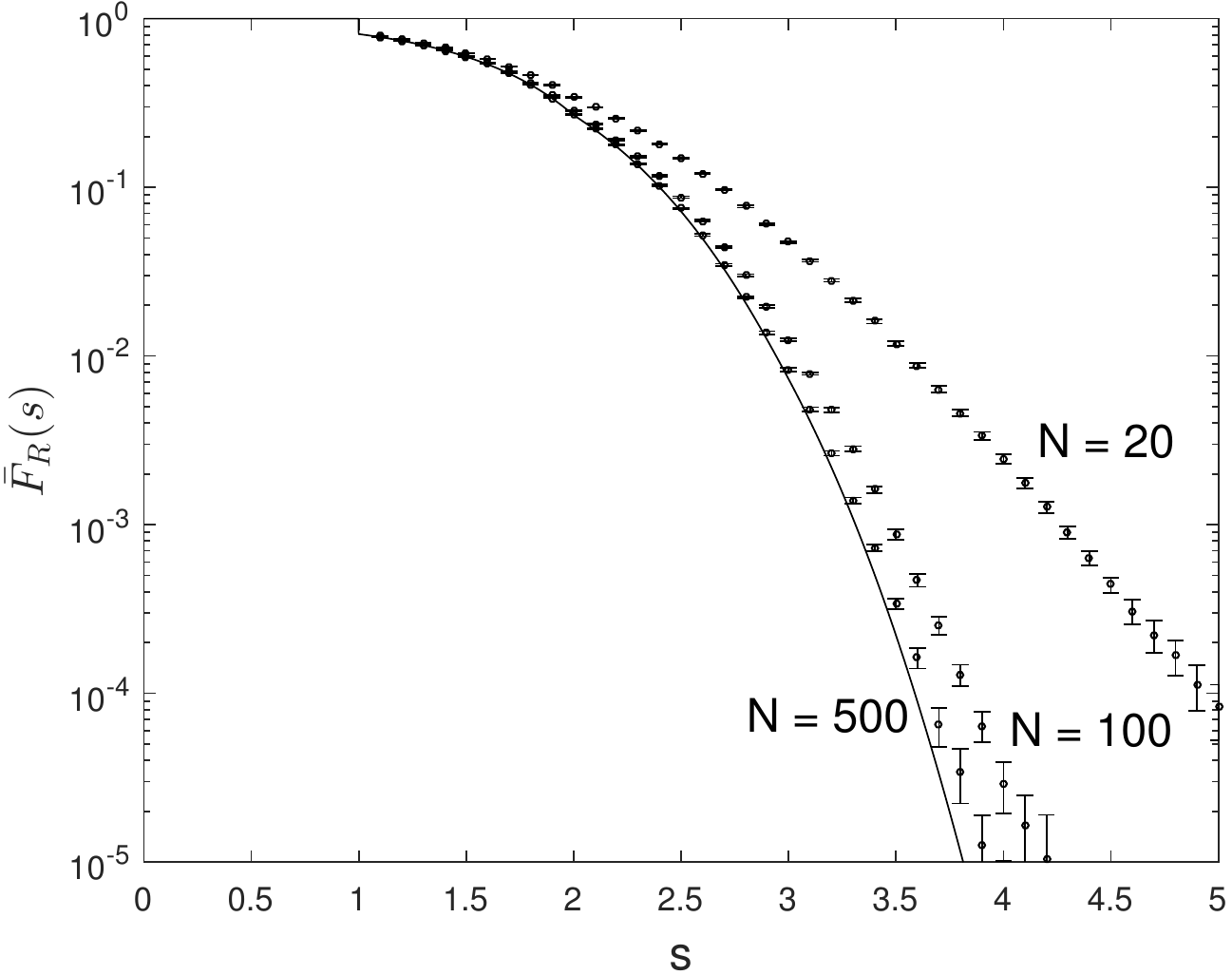}
\caption{Deterministic size one jobs, $d = 2$, $\lambda = 0.9$}
\label{fig:validateDET}
\end{subfigure}
\caption{Limiting response time distribution vs.~simulation 
for $N$ servers with power law and deterministic job sizes with mean $1$.
The full line represents the limiting response time distribution.}
\label{fig:validatePOWDET}
\end{figure*}

Figure \ref{fig:validateEXP} compares the expression for the limiting response time distribution given by
\eqref{eq:FR} for exponential job sizes with simulation experiments. In the simulation the number of servers equals $N=100$
servers, the $95\%$ confidence intervals are computed based on $10$ runs that each start from an empty system. The agreement with simulation is very good
(except for high loads combined with a small $d$) considering that we are simulating a system with only $100$ servers.

In Figure \ref{fig:validateEXP2} we look at the impact of the number of simulated servers $N$ under high loads
when $d=2$. We note that the limiting distribution is not necessarily a good match for the tail probabilities of the response time
when $N$ is small, e.g., $N=20$, but the accuracy quickly improves as the number of servers increases.

In Figure \ref{fig:validateHEXP} and \ref{fig:validateHEXP2} we look at a similar setting as in Figure \ref{fig:validateEXP}, but the job sizes
now follow a hyperexponential distribution with $f=1/2$ (see Section \ref{sec:SCV} for details).  
In this case the $95\%$ confidence intervals are computed based on $25$
runs. We note that even though the job sizes are now
substantially more variable, the accuracy seems quite similar to the exponential case. Thus, more variable 
job size distributions do not necessarily imply worse accuracy for a fixed $N$.

Figure \ref{fig:validatePOWDET} illustrates the accuracy of the limiting
response time distribution in case of power law and deterministic job sizes
(computed via the fixed point iteration in Section \ref{sec:fixed}).
More specifically, for the power law distribution we used $\bar{G}(s) = s^{-\beta}$ with $\beta=2$.
This implies that the mean job size is finite and equal to $2$, while the variance of the job size distribution
is infinite. In the deterministic case the job size equals $1$. 
The figure indicates that somewhat larger $N$ values are needed to closely match 
the limiting response time distribution compared to the (hyper)exponential case.


\section{Conclusions}\label{sec:concl}
In this paper we studied the limiting workload and response time distribution 
of the LL(d) policy which assigns an incoming job to a server with the least work left
among $d$ randomly selected servers. We introduced a fixed point iteration to determine the
limiting workload distribution for general service time distributions and any non-idling
service discipline and studied its convergence. We derived a closed form expression for
both the workload and response time distribution in case of exponential job sizes and indicated
that these distributions can be computed easily by solving a set of ordinary differential equations
for phase-type distributed job sizes. 

We provided insight into the gains that can be expected when exact
workload information is used instead of the coarser queue length information by comparing the performance of the LL(d) policy 
with the classic SQ(d) policy. Such a comparison is 
relevant to understand the performance gains offered by schedulers implementing {\it late binding}. 
In  this regard we demonstrated that 
late binding offers significant gains over SQ(d) for a wide range of arrival rates, even when taking
the late binding overhead into account.

\bibliographystyle{ACM-Reference-Format}
\bibliography{../../../PhD/thesis}

\end{document}